\documentclass[twoside,onecolumn]{article}

\usepackage[english]{babel} % Language hyphenation and typographical rules

\usepackage[hmarginratio=1:1,top=32mm,columnsep=20pt]{geometry} % Document margins
\usepackage[hang, small,labelfont=bf,up,textfont=it,up]{caption} % Custom captions under/above floats in tables or figures
\usepackage{booktabs} % Horizontal rules in tables

\usepackage{geometry} % to change the page dimensions
\usepackage{graphicx} % support the \includegraphics command and options

\usepackage{booktabs} % for much better looking tables
\usepackage{array} % for better arrays (eg matrices) in maths
\usepackage{paralist} % very flexible & customisable lists (eg. enumerate/itemize, etc.)
\usepackage{verbatim} % adds environment for commenting out blocks of text & for better verbatim
\usepackage{subfig} % make it possible to include more than one captioned figure/table in a single float
\usepackage{amsmath}	% split equations, equation numbering
\allowdisplaybreaks[1]	% page break, but avoid when not necessary
\numberwithin{equation}{section}	% numbering equations with section
\usepackage{amsfonts}	% math symbols
\usepackage{amsthm}	% theorems, definitions, examples, ...
\usepackage{cases}
\usepackage{bbm}		% indicator function
\usepackage{braket}	% sets, brakets, ...
\usepackage{mathtools}	% norm, ...
\usepackage[english]{varioref}	% for \vref (reference + page number)
\usepackage[multiple]{footmisc}	% for multiple footnotes

\usepackage{fancyhdr} % This should be set AFTER setting up the page geometry
\theoremstyle{plain}
\newtheorem{example}{Example}[section]
\theoremstyle{plain}
\newtheorem{remark}{Remark}[section]
\theoremstyle{plain}
\newtheorem{prop}{Proposition}[section]
\theoremstyle{plain}

\theoremstyle{plain}
\newtheorem{theorem}{Theorem}[section]
\theoremstyle{plain}
\newtheorem{corollary}{Corollary}[section]
\theoremstyle{plain}
\newtheorem{definition}{Definition}[section]
\theoremstyle{assumption}
\newtheorem{assumption}{Assumption}[section]
\theoremstyle{plain}

\theoremstyle{plain}
\newtheorem{notation}{Notation}[section]

\DeclarePairedDelimiter{\abs}{\lvert}{\rvert}

\usepackage{sectsty}
\allsectionsfont{\sffamily\mdseries\upshape} % (See the fntguide.pdf for font help)
\usepackage{abstract} % Allows abstract customization
\usepackage{titlesec} % Allows customization of titles
\titleformat{\section}[block]{\large\bfseries\centering}{\thesection.}{1em}{} % Change the look of the section titles
\titleformat{\subsection}[block]{\large\bfseries}{\thesubsection.}{1em}{} % Change the look of the section titles

\usepackage[nottoc,notlof,notlot]{tocbibind} % Put the bibliography in the ToC
\usepackage[titles,subfigure]{tocloft} % Alter the style of the Table of Contents

\usepackage[a-1b]{pdfx} % per generare un file conforme alla norma PDF/A-1b
\usepackage{titling} % Customizing the title section
\usepackage{hyperref}

\hypersetup{%
    pdfpagemode={UseOutlines},
    bookmarksopen,
    pdfstartview={FitH},
    colorlinks,
    linkcolor={black},
    citecolor={black},
    urlcolor={black}
  }

\pretitle{\begin{center}\Huge\bfseries} % Article title formatting
\posttitle{\end{center}} % Article title closing formatting
\title{A Stochastic Control Approach to Public Debt Management} % Article title
\author{%
\textsc{Brachetta M.} \thanks{Department of Mathematics, Politecnico of Milan, Piazza Leonardo da Vinci, 32 - 20133 Milan. Corresponding author. ORCID: 0000-0001-6940-4749.} \\[1ex]
\normalsize \href{mailto:matteo.brachetta@unich.it}{matteo.brachetta@polimi.it}
\and
\textsc{Ceci, C.}\thanks{Department of Economics, University of Chieti-Pescara, Viale Pindaro, 42 - 65127 Pescara, Italy. ORCID: 0000-0003-2796-6588.}\footnotemark[1] \\[1ex]
\normalsize \href{mailto:c.ceci@unich.it}{c.ceci@unich.it}
}
\date{} % Leave empty to omit a date

\providecommand{\keywords}[1]{\textbf{\textit{Keywords:}} #1}
\providecommand{\jelcodes}[1]{\textbf{\textit{JEL Classification codes:}} #1}

\begin{document}
\maketitle

\begin{abstract}
\noindent 
We discuss a class of debt management problems in a stochastic environment model.  We propose a model for the debt-to-GDP (Gross Domestic Product) ratio  where the government interventions via fiscal policies affect the public debt and the GDP growth rate at the same time. We allow for stochastic interest rate and possible correlation with the GDP growth rate through the dependence of both the processes (interest rate and GDP growth rate) on a stochastic factor which may represent any relevant macroeconomic variable, such as the state of economy. We tackle the problem of a government whose goal is to determine the fiscal policy in order to minimize a general functional cost. We prove that the value function is a viscosity solution to the Hamilton-Jacobi-Bellman equation and provide a Verification Theorem based on classical solutions. We investigate the form of the candidate optimal fiscal policy in many cases of interest, providing interesting policy insights. Finally, we discuss two applications to the debt reduction problem and debt smoothing, providing explicit expressions of the value function and the optimal policy in some special cases.
\end{abstract}

\noindent\keywords{Optimal stochastic control; government debt management; optimal fiscal policy; Hamilton-Jacobi-Bellman equation.}\\
\noindent\jelcodes{C02, C61, H63, E62.}\\
%% C02 = Mathematical Methods; C61 	Optimization Techniques • Programming Models • Dynamic Analysis 
%% H63 	Debt • Debt Management • Sovereign Debt; E62 	Fiscal Policy
%% G220 Insurance; Insurance Companies; Actuarial Studies
%% C610 Optimization Techniques; Programming Models; Dynamic Analysis
%% G110 Portfolio Choice; Investment Decisions
%\noindent\msccodes{93E20, 60J60, 91B64.}\\
% 93E20 (stochastic control); 60J60 (diffusion process); 91B64 (macroeconomics theory…).
%% 93E20 Optimal stochastic control, 91B30 Risk theory, insurance, 91G80 Financial applications of other theories (stochastic control, calculus of variations, PDE, SPDE, dynamical systems), 60G35 Signal detection and filtering (nella sezione PROBABILITY THEORY AND STOCHASTIC PROCESSES), 60G57 Random measures, 60J75 Jump processes, 93E11 Filtering (nella sezione SYSTEMS THEORY; CONTROL)
%\noindent \textit{Declarations of interest: none.}

%--------------------------------------------------------------------------------
%	INTRO
%--------------------------------------------------------------------------------

\section{Introduction}

Public debt management is a wide and complex topic in Economics. On the one hand, it is recognized the important role played by public debt in welfare improving, for example as a tax smoothing tool (starting from the seminal paper \cite{barro79}) or as savings absorber (see \cite{domar1944} and \cite{necksturm2008} among others). On the other hand, the current high levels of debt in some developed countries has drawn the attention of many economists, especially because of possible effects on future taxation levels (see the quoted papers \cite{domar1944} and \cite{necksturm2008}).

The question of debt management, in broad sense, is up for discussion essentially because ``we simply do not have a theory of the optimum debt level'' (see \cite{Wyplosz2005}). However, in the very recent years some authors gave a rigorous mathematical formulation of a debt management problem, namely the optimal debt ceiling (that is the debt-to-GDP  (Gross Domestic Product) ratio level at which the government should intervene in order to reduce it), see \cite{CA2016, CA2018, ferrari2018, ferrari2020, CCF2020}.

These works focus especially on the debt reduction problem of developed countries, where the government aims at reducing the debt-to-GDP ratio through the minimization of two costs: the cost of holding debt and the intervention cost (reducing public spending or increasing taxes).  This study is motivated by the fact that concurrently with the financial crisis started in 2007, the debt-to-GDP ratios exploded.

In \cite{CA2016, CA2018, ferrari2018, CCF2020} the possibility for the government to increase the level of debt ratio is precluded (fiscal deficit is not allowed) and any possible benefit resulting from holding debt is neglected. In practice, sometimes debt reduction policies might not be appropriate, since public investments in infrastructure, healthcare, education and research induce social and financial benefits (see \cite{Ostryetal}). In \cite{ferrari2020} policies of debt reduction and debt expansion are accounted by modeling controls of bounded variation and introducing, in addition to the cost of reduction policies, a benefit associated to expansionary policies.

Moreover, it is a well known evidence that 
policies of debt expansion (deficit) induce also an increase of the GDP growth rate of the country (which in turn could imply a reduction of the debt ratio) and this phenomenon is not exploited in the existing mathematical literature on the topic. This paper wishes to be a first effort to fill this lack.\\
%In \cite{ferrari2018} the government’s debt management problem is discussed as a singular stochastic control problem where the government 's objective is the reduction of  the debt ratio through the minimization of two opposing costs, the cost of having debt on the one hand, and that from the reduction policy on the other hand. It is proved that the solution of the control problem is related to that of an auxiliary optimal stopping problem and that the optimal policy is found to be that of keeping the debt ratio under an inflation-dependent ceiling. In \cite {CCF2020} a similar problem the problem of a government aiming at reducing the debt ratio under partial information where the underlying macroeconomic conditions are not directly observed. ...can reduce the level of the debt-to-GDP ratio by adjusting the primary budget balance explores the case of a government whose objective is that of reducing the debt ratio through the minimization of two opposing costs, namely the expected opportunity cost of having debt on the one hand, and the expected cost from the reduction policy on the other hand. (e.g. through fiscal interventions like raising taxes or reducing expenses and aims to ...The link to optimal stopping is then exploited to characterise the optimal debt reduction policy.

We provide a rigorous mathematical formulation for a class of debt management problems, which are modeled as stochastic control problems. In particular, we tackle the problem of a government whose goal is to determine the fiscal policy in order to minimize a general functional cost, which depends on the debt-to-GDP ratio, an external driver (e.g. the state of economy) and the fiscal policy itself.
 
The main improvement of our debt-to-GDP model is that the GDP growth rate depends on the fiscal policy.  In classical models, the GDP growth rate is assumed to be constant, see for instance  \cite{CA2016, CA2018} and references therein. In \cite{ferrari2018, ferrari2020, CCF2020} it is allowed to be a stochastic process, which may be modulated by an unobservable continuous time Markov chain as in \cite{CCF2020}. However, in all these models the government's interventions via fiscal measures do not affect the GDP growth rate, 
any policy of surplus decreases the debt ratio (see \cite{CA2016, CA2018, ferrari2018, CCF2020}), whereas any deficit policy increases  it (as in \cite{ferrari2020}). In the reality, the effects of the fiscal policy on the debt-to-GDP ratio is more complex.
 
For example, we might assume that the GDP growth rate decreases when the government increases its surplus (which is mostly the case in normal conditions). In this scenario, the gross debt will decrease as well, but the final effect on debt-to-GDP ratio is not unidirectional, contrary to classical models and the above quoted papers, where any surplus translates into debt ratio reduction. Indeed, the economics literature recognizes the possibility of achieving debt reduction via deficit policies, see for instance \cite{domar1944,delong2012,fatas2018}.

Moreover, we assume that both the interest rate on debt and the GDP growth rate are affected by a stochastic factor $Z$, which may represent the state of economy, the economic outlook, or any other macroeconomic variable.  The presence of this external driver $Z$ induces a correlation between the GDP growth rate and the interest rate, which is a well known phenomenon. Furthermore, we assume that $Z$ and the GDP are driven by two correlated Brownian motion, so that we introduce an additional type of dependence. In practice, the macroeconomic conditions described by $Z$ and the GDP have a common source of uncertainty. %For instance, when $Z$ is the state of economy of trading partners, then $Z$ affects the GDP growth rate and, at the same time, partners usually have common sources of uncertainty.

In this framework we solve the problem of minimizing the expected total cost over an infinite time horizon. As mentioned above, the functional cost depends on the dynamics of the debt-to-GDP ratio, the exogenous factor $Z$ and the fiscal policy. This is a very flexible and general problem formulation, which includes many debt management problems as special cases. For instance, we explore the applications to debt reduction and debt smoothing, which is the minimization of the distance between the current debt and a given threshold.\\

The mathematical contributions of the paper are the following. We rigorously derive the debt ratio dynamics, which is controlled by the government's interventions (which include surplus and deficit strategies) and formulate the arising stochastic control problem. Under general assumptions on the functional cost, we prove that the value function is a continuous viscosity solution to the Hamilton-Jacobi-Bellman (HJB) equation and provide a Verification Theorem which applies whenever the HJB equation has a classical solution. Next, we investigate the structure of the candidate optimal fiscal policy, discussing some cases of interest. In particular, when the functional cost is increasing on the debt ratio and does not explicitly depend by the control, we find that the optimal strategy depends only on the effect of the fiscal policy on the GDP growth rate, not by the current level of debt-to-GDP. A similar result was obtained in a different context by \cite{Annicchiarico_etal}.
%on the structure of the function which models the effect of the control on the GDP growth rate. 
Thus it becomes crucial understanding how the government's interventions via fiscal policies can influence the GDP growth rate of the country. To find more explicit solutions, we consider the example of a linear impact.  Finally, we discuss two applications, namely debt reduction and debt smoothing 
and provide a numerical simulation for the second one. \\

The paper is organized as follows.  In Section \ref{section:model} we propose and motivate our model formulation. In Section \ref{section:pb} we illustrate the optimization problem. In Section \ref{properties} we provide some properties of the value function. In Section \ref{characterization} we prove that the value function is 
a solution in the viscosity sense to the Hamilton-Jacobi-Bellman equation and we provide a Verification Theorem based on classical solution. In Section \ref{optimal_stra} we discuss the minimization problem involved in the HJB equation. Finally, some special cases of interest are discussed in Section \ref{Applications}. %, namely debt reduction and debt smoothing.

%--------------------------------------------------------------------------------
%	MODEL FORMULATION
%--------------------------------------------------------------------------------

\section{Model formulation}
\label{section:model}

We propose a stochastic model for the gross public debt and the gross domestic product (GDP) in presence of correlation between the interest rate on debt and the economic growth. Moreover, the government interventions through fiscal policies affect the GDP growth rate and the public debt at the same time.

Consider a complete probability space $(\Omega,\mathcal{F},\mathbb{Q},\mathbb{F})$ endowed with a complete and right continuous filtration $ \mathbb{F} := \{\mathcal{F}_t\}_{t \geq 0}$. Let $D = \{D_t\}_{t\geq0}$  denote the gross public debt process and $Y = \{Y_t\}_{t\geq 0}$ the gross domestic product process of a country.
According to classical models in economics literature (see \cite{Blanchard85} among others), the sovereign debt stock evolves as:
\[
dD_t = rD_t\,dt - d\xi_t \;, \quad D_0>0 \;,
\]
where $r>0$ denotes the real interest rate on debt and $\xi$ is the fiscal policy, with the convention that positive values correspond to primary surplus, while negative values represent deficit.  
We extend the model by introducing  a stochastic interest rate of the form $r_t= r(Z_t)$, with $r$ a positive measurable function and $Z=\{Z_t\}_{t>0}$ a stochastic factor described by the following stochastic differential equation (SDE):
\begin{equation}
\label{eqn:Z}
dZ_t = b_Z(Z_t)\,dt + \sigma_Z(Z_t)\,dW^Z_t \;, \qquad Z_0=z\in\mathbb{R} \;,
\end{equation}
where $W^Z = \{ W^Z_t\}_{t \geq 0}$ is a standard Brownian motion. We assume existence and strong uniqueness of the solution to the SDE \eqref{eqn:Z}. %Moreover, $\sigma_Z$ is a continuous function. 
The process $Z$ describes any stochastic factor, such as underlying macroeconomic conditions, which affects the interest rate on debt of the country. 

Let $u = \{u_t\}_{t\geq0}$ be the rate of the primary balance expressed in terms of the debt, that is $d\xi_t= u_t D_t dt$. We assume that the real GDP growth rate at time $t$ is of the form  $g(Z_t,u_t)$, with $g(z,u)$ being a measurable function of its arguments. Precisely, the pair $(D,Y)$ follows  
\begin{equation}\label{sistema}
\begin{cases}
	dD_t = D_t(r(Z_t)-u_t)\,dt \;, \quad D_0>0 \;,
	\\
	dY_t = Y_t(g(Z_t,u_t)\,dt +\sigma\,d\widetilde W_t) \;, \quad Y_0>0 \;,
\end{cases}
\end{equation}
where $\widetilde W=\{\widetilde W_t\}_{t>0}$ a standard Brownian motion correlated with $W^Z$ and $\sigma>0$ is the GDP volatility.

The debt-to-GDP ratio $X=\{X_t = \frac{D_t}{Y_t} \}_{t\ge0}$ dynamics can be derived by It\^o's formula:
\[
\begin{split}
dX_t &= \frac{1}{Y_t} \big (r(Z_t)- u_t \big ) D_t dt -\frac{D_t}{Y_t^2}Y_t(g(Z_t,u_t)\,dt + \sigma\,d\widetilde W_t) 
+ \frac{D_t}{Y_t^3}\sigma^2Y_t^2\,dt \\
&= X_t \big (r(Z_t) -g(Z_t,u_t) -u_t\big ) \,dt + X_t\sigma(\sigma\,dt-d\widetilde W_t)   \;, \qquad X_0=\frac{D_0}{Y_0} \;.
\end{split}
\]
Now we can introduce a new measure $\mathbb{P}$, equivalent to $\mathbb{Q}$, such that 
$W=\{W _t\}_{t\geq0}= \{\sigma t - \widetilde W_t\}_{t\geq0}$ is a $\mathbb{P}$-Brownian motion and $W^Z$ remains a $\mathbb{P}$-Brownian motion.

Hence, under $\mathbb{P}$, the debt ratio $X^{u,x}=\{X^{u,x}_t\}_{t\geq0}$ is a controlled process which solves the following  SDE:
\begin{equation}
\label{eqn:debt}
dX^{u,x}_t = X^{u,x}_t [(r(Z_t)-g(Z_t,u_t)-u_t)\,dt + \sigma\,dW_t] \;, \qquad X^{u,x}_0=x \;,
\end{equation}
where $x>0$ is the initial debt-to-GDP ratio and the control $u= \{u_t\}_{t\geq0}$ denotes the fiscal policy (i.e. the ratio of primary surplus to gross debt).

Let us remark some aspects of our model, which is different from those usually introduced in the existing literature:
\begin{enumerate}
	\item This model extends the one considered in \cite{CA2016} and \cite{CA2018} where $r$ and $g$ are constant.  Our main improvement is that the GDP growth rate is now affected by the fiscal policy. For instance, if we take a function $g$ decreasing in $u$, we capture a well known effect: when the government generates surplus, the debt stock is reduced and, at the same time, the GDP is so. Hence the final effect on the debt-to-GDP ratio is not unidirectional as in the models proposed in \cite{CA2016, CA2018, ferrari2018}  and \cite{CCF2020}, where any surplus translates to debt-to-GDP reduction. 
	\item Another property of our model is that the interest rate and the GDP  growth rate have a common source of uncertainty, which is modeled through the external driver $Z$. In many applications $Z$ could represent the economic environment, the state of the economy, the economic outlook.  
The pathwise measurement of covariance between the interest rate and the GDP  growth rate  is given by the covariation between the two processes. When the functions $r$ and $g$ are sufficiently regular w.r.t. $z \in \mathbb R$, from It\^o's formula it  can be  computed   for any fixed $u$:
$$< r(Z_t), g(Z_t,u)> = \int_0^t \frac{\partial{g}}{\partial{z}}(Z_s, u) r'(Z_s) \sigma^2_Z(Z_s) ds \;.$$
In particular, this implies that 
$$\mathbb{E}[r(Z_t)g(Z_t,u)] =\mathbb{E}[ \int_0^t \frac{\partial{g}}{\partial{z}}(Z_s, u) r'(Z_s) \sigma^2_Z(Z_s) ds ] \;.$$
	\item Moreover, our model takes into account possible correlation between the Brownian motions $W^Z$ and $W$ driving the dynamics of the environmental stochastic factor process $Z$ and the debt-to-GDP process $X$, respectively. For instance, there could be a common source of uncertainty between debt-to-GDP and macroeconomics conditions.
\end{enumerate}

%When the Brownian motion is removed, this formula reduces to the simplest model formulation of the debt-to-GDP ratio in continuous time. However, it is well known that this model neglects some important aspects. On the one hand, the GDP growth rate is affected by the fiscal policy. Precisely, the fiscal surplus/deficit being a component of the GDP, a positive surplus translates to GDP reduction. On the other, there exists a correlation between the interest rate on debt and the GDP.\\

In the sequel we assume the following hypotheses.

\begin{assumption}
\label{ass:1}
The following assumptions are required in the sequel:
\begin{itemize}
\item $r\colon\mathbb{R}\to(0,+\infty)$, i.e. the interest rate on debt, is such that $r(z)\le R$ $\forall z\in\mathbb{R}$ for a given constant $R>0$; %\ourcomment{valutare se $r<0$}
\item $u_t$, i.e. the surplus (or deficit)-to-debt ratio at time $t>0$, takes values in a compact set $[-U_1,U_2]$, where $U_1>0$ denotes the maximum allowed deficit-to-debt ratio and $U_2>0$ is the maximum surplus-to-debt; 
\item $g\colon\mathbb{R}\times[-U_1,U_2]\to\mathbb{R}$, which represents the GDP growth rate, is bounded by
\begin{equation}
\label{g_bounded}
\bar{g}_1 \le g(z,u) \le \bar{g}_2 \qquad \forall (z,u)\in\mathbb{R}\times[-U_1,U_2] \;,
\end{equation}
for some suitable constants $\bar{g}_1<0<\bar{g}_2$.
\end{itemize}
\end{assumption}

\begin{example} \label{example:g}   
We consider the case where the fiscal policy has a linear impact on the GDP growth rate, precisely

\begin{equation} \label{linear1}
g(z,u) = g_0(z) - \alpha(z) u \;,
\end{equation}
with $g_0\colon\mathbb{R}\to\mathbb{R}$, $\alpha \colon\mathbb{R}\to (0, + \infty)$ measurable and bounded functions.  
The process $\{g_0(Z_t)\}_{t\geq 0}$ is the GDP growth rate when no intervention is considered and depends on the environment stochastic factor $Z$. 
Any positive intervention $u_t >0$  leads to a reduction of the GDP growth rate $g(Z_t,u_t)$. Conversely, any negative intervention $u_t <0$ leads to an increase of the GDP growth rate $g(Z_t,u_t)$. Both the effects are modulated by the environment stochastic factor $Z$ via the coefficient $\alpha(Z_t)$. 
A simplified model can be obtained with $\alpha(z)=\alpha \neq 1$ $\forall z \in \mathbb{R}$, that is  
\begin{equation}
\label{linear}
g(z,u) = g_0(z) - \alpha u \;.
\end{equation}
In this special case, equation \eqref{eqn:debt} reduces to
\begin{equation}
\label{eqn:debt1}
dX^{u,x}_t = X^{u,x}_t [(r(Z_t)- g_0(Z_t ) - (1 -\alpha) u_t)\,dt + \sigma\,dW_t] \;, \qquad X^{u,x}_0=x \;.
\end{equation}
Clearly, $0< \alpha <1$ means that the effect of the government fiscal policy with $u_t >0$  (surplus)  leads to a reduction of the instantaneous debt-to-GDP  growth rate, while $\alpha>1$ leads to a reduction of the instantaneous debt-to-GDP growth rate when $u_t <0$ (deficit). This occurs because, for $0< \alpha <1$, the government fiscal policy has a smaller effect on the GDP growth rate w.r.t. the debt growth rate (see the first equation in \eqref{sistema}), while for $\alpha >1$  the fiscal policy has a larger effect on the GDP growth rate thank on debt growth rate. 
\end{example}

For the sake of notation simplicity we make use of the following convention.
\begin{notation}
The dynamics of $X^{u,x}$, see equation \eqref{eqn:debt}, must be assigned together with equation \eqref{eqn:Z}. However, when the expectations are considered we will neglect the conditioning to $Z_0=z$, so that the two notations below are equivalent for us:
\[
\mathbb{E}[X^{u,x}_t \mid Z_0=z] = \mathbb{E}[X^{u,x}_t]  \;.
\]
We will use the left hand notation either to specify a different initial value of $Z$ or to emphasize the dependence of $Z$.
\end{notation}

 \begin{remark}\label{remark:Xdiseq}
Let us observe that the SDE \eqref{eqn:debt} admits an explicit solution:
\begin{equation}
\label{eqn:Xsol}
X^{u,x}_t = xe^{\int_0^t (r(Z_s)-g(Z_s,u_s)-u_s)\,ds - \frac{1}{2}\sigma^2t}e^{\sigma W_t} \quad \forall t \geq 0, \;  \mathbb{P}-a.s. \;.
\end{equation}
Thus by  Assumption \ref{ass:1}, for any $m>0$ we have that
\[
x^m e^{-m(\bar{g}_2 + U_2) t} e^{- \frac{m}{2}\sigma^2t + m\sigma W_t} \le 
(X^{u,x}_t)^m \le x^m e^{m(R - \bar{g}_1 + U_1) t} e^{- \frac{m}{2}\sigma^2t + m\sigma W_t} \quad \forall t \geq 0, \;  \mathbb{P}-a.s.
\]
from which we get this estimation:
\begin{equation}
\label{eqn:EXm_finite}
\mathbb{E}[(X^{u,x}_t)^m] \le x^m e^{\lambda_m t} \quad \forall t \geq 0\;,
\end{equation}
where
\begin{equation}
\label{eqn:lambda}
\lambda_m \doteq m(R - \bar{g}_1 + U_1) + m(m-1)\frac{\sigma^2}{2} \;,
\end{equation}
\end{remark}
 
\begin{remark}
We can show that the condition $r(z)-g(z,0)\ge G$ $\forall z\in\mathbb{R}$, for some constant $G>0$, implies an explosive debt-to-GDP ratio when no intervention is considered. Indeed, by equation \eqref{eqn:Xsol} we have that
\begin{equation} \label{eqn:X0}
X^{0,x}_t = xe^{\int_0^t (r(Z_s)-g(Z_s,0))\,ds}e^{- \frac{1}{2}\sigma^2t + \sigma W_t} \ge xe^{Gt}M_t \end{equation}
where $M_t = e^{- \frac{1}{2}\sigma^2t + \sigma W_t}$ denotes the well known exponential martingale with  $\mathbb{E}[M_t]=1$ $\forall t \geq 0$. Hence 
\[
\lim_{t \to + \infty} \mathbb{E}[X^{0,x}_t] \geq  \lim_{t \to + \infty} xe^{Gt} = + \infty \;,
\]
which implies $\lim_{t \to + \infty} X^{0,x}_t = + \infty$ $\mathbb{P}-a.s.$.
However, when $r$ and $g$ are constant functions, the debt explodes if $r-g>0$, which is a popular result in economics literature.
\end{remark}

An important feature for a country is to apply fiscal policies which are sustainable. An explosive debt-to-GDP ratio is not a problem for a country if the discounted debt-to-GDP ratio w.r.t. the  interest rate on debt converges to $0$. The following definition is standard in Economics (see e.g. \cite{FINCKE2011}).
\begin{definition}
A fiscal policy $u$ is called sustainable if it realizes
\begin{equation}
\label{eqn:sustainable}
\lim_{t\to+\infty}e^{-\int_0^t r(Z_s)\,ds}X^{u,x}_t =0 \qquad \mathbb{P}-a.s.\;,
\end{equation}
see e.g. \cite{Blanchard85}.
\end{definition}

%\begin{remark}If $g(z,0) \geq c $ $\forall z\in\mathbb{R}$, for some constant $c>0$, then the no-intervention policy ($u_t=0$, $\forall t \geq 0$) is sustainable.This follows because we have that\[e^{-\int_0^t r(Z_s)\,ds} X^{0,x}_t  = xe^{-\int_0^t  g(Z_s,0)\,ds}e^{- \frac{1}{2}\sigma^2t + \sigma W_t} \leq  xe^{-ct} M_t \]hence\[0 \leq \lim_{t \to + \infty} \mathbb{E}[e^{-\int_0^t r(Z_s)\,ds} X^{0,x}_t ] \leq  \lim_{t \to + \infty} xe^{- ct} = 0,\]which finally  implies  \eqref{eqn:sustainable}.

%Performing similar computations as in  Remark \ref{remark:Xdiseq} we get that the no-intervention policy ($u_t=0$, $\forall t \geq 0$) is sustainable if the GDP growth rate without  intervention satisfies $g(z,0) \geq c >0$ $\forall z\in\mathbb{R}$, for some constant $c>0$.
%\end{remark}

\begin{remark}\label{sustU}
 The bounds $U_i$, $i=1,2$, depend by structural economic and political characteristics of the country, in general.  It seems reasonable to choose the maximum level of  deficit and surplus, $U_i$, $i=1,2$, respectively, by imposing that the fiscal policies identically equal to these maximum levels (i.e. $u^1_t = - U_1$ and $u^2_t = U_2$ $\forall t \geq 0$) turn out to be sustainable for the country, i.e.
 \begin{equation}
\lim_{t\to+\infty}e^{-\int_0^t r(Z_s)\,ds}X^{-U_1,x}_t =0 \;,  \quad  \lim_{t\to+\infty}e^{-\int_0^t r(Z_s)\,ds}X^{U_2,x}_t =0 \quad \mathbb{P}-a.s. \;.
\end{equation}
For instance, in case of linear GDP growth rate as in \eqref{linear1}, we get that 
$$\mathbb{E}[ e^{-\int_0^t r(Z_s)\,ds}X^{-U_1,x}_t ] = x \mathbb{E}[ e^{-\int_0^t (g_0(Z_s) - U_1(1 -\alpha)) \,ds} M_t] $$
and
$$\mathbb{E}[ e^{-\int_0^t r(Z_s)\,ds}X^{U_2,x}_t ] = x \mathbb{E}[ e^{-\int_0^t (g_0(Z_s) + U_2(1- \alpha)) \,ds} M_t]$$
(here we recall that $M_t = e^{- \frac{1}{2}\sigma^2t + \sigma W_t}$).
Let $\underline g_0= \min_{z \in \mathbb R}  g_0(z)$ we get 
$$\mathbb{E}[ e^{-\int_0^t r(Z_s)\,ds}X^{-U_1,x}_t ] \leq x e^{- (\underline g_0 - U_1(1 -\alpha))t} \; , $$
$$\mathbb{E}[ e^{-\int_0^t r(Z_s)\,ds}X^{U_2,x}_t ] \leq x e^{- (\underline g_0 + U_2(1- \alpha))t} \; , $$
which imply that if $U_1$ and $U_2$ satisfy
$$U_1(1-\alpha) < \underline g_0, \quad U_2(\alpha -1) < \underline g_0 \; , $$
both the fiscal policies $u^1_t = - U_1$ $\forall t \geq 0$ and $u^2_t = U_2$ $\forall t \geq 0$ are sustainable.
The economic interpretation is clear: when the maximum deficit positive impact on GDP growth is greater than the effect on debt, then $-U_1$ is sustainable. Similarly, when the negative effect of the maximum surplus on GDP is surpassed by the positive effect on debt reduction, then $U_2$ becomes sustainable.
\end{remark}

%--------------------------------------------------------------------------------
%	PROBLEM FORMULATION
%--------------------------------------------------------------------------------

\section{Problem formulation}
\label{section:pb}

Let us recall the dynamics of debt-to-GDP ratio:
\[
\begin{cases}
	dX^{u,x}_t = X^{u,x}_t [(r(Z_t)-g(Z_t,u_t)-u_t)\,dt + \sigma\,dW_t] \;, & X^{u,x}_0=x \;,
	\\
	dZ_t = b_Z(Z_t)\,dt + \sigma_Z(Z_t)\,dW^Z_t \;, & Z_0=z\in\mathbb{R} \;.
\end{cases}
\]

We denote by $\rho\in[-1,1]$ the correlation coefficient between $W$ and $W^Z$.\\

We consider the problem of a government which wants to choose the fiscal policy in order to optimally manage the sovereign debt. For this purpose we introduce the class of admissible fiscal strategies.

\begin{definition}[Admissible fiscal policies]
\label{def:admissible}
We denote by $\mathcal{U}$ the family of all the $\mathbb{F}$-predictable and $[-U_1,U_2]$-valued processes $u=\{u_t\}_{t\ge0}$. 
%When we want to emphasize the dependence on the initial Debt-to-GDP ratio $x>0$, we will write $\mathcal{U}$. 
%such that
%\begin{equation}
%\label{eqn:admissiblecond}
%\mathbb{E}\biggl[ \int_0^{+\infty} e^{-\lambda t}f(X^{u,x}_t)\,dt \biggr] <+\infty \;,
%\end{equation}
\end{definition}

The main goal will be to minimize the following objective function:
\begin{equation}
\label{eqn:J}
J(x,z,u) = \mathbb{E}\biggl[ \int_0^{+\infty} e^{-\lambda t}f(X^{u,x}_t, Z_t, u_t)\,dt \mid Z_0=z \biggr],
\qquad (x,z)\in(0,+\infty)\times\mathbb{R}, u\in\mathcal{U} \;,
\end{equation}
where $\lambda>0$ is the government discounting factor and $f\colon(0,+\infty)\times\mathbb{R}\times[-U_1,U_2]\to[0,+\infty)$ is a cost function satisfying the following hypotheses.

\begin{assumption}
\label{ass:f}
We assume that %$\lambda>\lambda_m$ (see equation \eqref{eqn:lambda}) and the cost function $f$ fulfills these conditions:
\begin{itemize}
\item $f$ is nonnegative;
%\item $f$ is increasing in $x\in(0,+\infty)$;
\item $\exists C>0$ and $m>0$ such that 
\[
f(x,z,u)\le C(1+x^m) \qquad \forall (x,z,u)\in(0,+\infty)\times\mathbb{R}\times[-U_1,U_2]\;,
\]
%and
%\[
%f(x,z,0)\le Cx^m \qquad \forall (x,z)\in(0,+\infty)\times\mathbb{R} \;,
%\]
\item $\lambda>\lambda_m$ (see equation \eqref{eqn:lambda}).
\end{itemize}
\end{assumption}

We observe that our problem formulation is very general and flexible. The cost function depends on the debt-to-GDP ratio, which has to be controlled, and the government can take into account fluctuations of the stochastic factor $Z$. For instance, when $Z$ represents the state of economy, countercyclical policies are allowed. In addition to this, the fiscal policy level can be explicitly controlled as well. Clearly, many operational problems can be addressed in this framework, depending on the configuration that the government assigns to the function $f$. We will investigate some applications in Section \ref{sec:applications}.\\

%\begin{example}
%\ourcomment{vedere se fare riferimento a CA}\\
%In \cite{CA2018} the authors assume
%\[
%f(x,u) = h(x) + ku \;,
%\]
%with $k>0$ and $h(x)=\alpha x^{m+1}$, where $m\in\{1,2,\dots\}$ denotes the aversion towards the debt.
%\end{example}

As announced, the government problem can be formalized in this way:
\begin{equation}
\label{eqn:infinitetime_pb}
v(x,z) = \inf_{u\in\mathcal{U}} J(x,z,u) \;, \qquad (x,z)\in(0,+\infty)\times\mathbb{R} \;.
\end{equation}

%\begin{remark}
%\label{remark:classUmonotone}
%Let us take $0<x<x'$ and $u\in\mathcal{U}(x')$. Since $X^{u,x}_t<X^{u,x'}_t$ $\forall t\ge0$ with probability $1$ and $f$ is increasing in $x$, then
%\[
%\mathbb{E}\biggl[ \int_0^{+\infty} e^{-\lambda t}f(X^{u,x}_t)\,dt \biggr] 
%< \mathbb{E}\biggl[ \int_0^{+\infty} e^{-\lambda t}f(X^{u,x'}_t)\,dt \biggr] < +\infty \;.
%\]
%By Definition \ref{def:admissible} this means that $u\in\mathcal{U}$, hence we conclude that 
%\[ \mathcal{U}\supseteq\mathcal{U}(x') \qquad \forall 0<x<x' \;. \]
%\end{remark}

\begin{prop}
\label{prop:Jfinite}
Every admissible strategy $u \in \mathcal{U}$ is such that $J(x,z,u)<+\infty$, $\forall (x,z)\in(0,+\infty)\times\mathbb{R}$.
\end{prop}
\begin{proof}
Under Assumptions \ref{ass:1} and  \ref{ass:f} and recalling \eqref{eqn:EXm_finite}, we get that $\forall u \in \mathcal{U}$, and  $(x,z)\in(0,+\infty)\times\mathbb{R}$
%\begin{equation}
%\begin{split}
%J(x,z,u)&=\mathbb{E}\biggl[ \int_0^{+\infty} e^{-\lambda t}f(X^{u,x}_t,Z_t,u_t)\,dt \biggr] \\
%&\le C\mathbb{E}\biggl[ \int_0^{+\infty} e^{-\lambda t}(X^{u,x}_t)^m\,dt 
%+ \int_0^{+\infty} e^{-\lambda t}\,dt \biggr]  \\
%&\le Cx^m  \int_0^{+\infty} e^{(\lambda_m-\lambda)t}\,dt + \frac{C}{\lambda} \\
%&= \frac{Cx^m}{\lambda-\lambda_m} + \frac{C}{\lambda} < +\infty \;.
%\end{split}
%\end{equation}
\begin{align}
J(x,z,u)&=\mathbb{E}\biggl[ \int_0^{+\infty} e^{-\lambda t}f(X^{u,x}_t,Z_t,u_t)\,dt \biggr] \notag\\
&\le C\mathbb{E}\biggl[ \int_0^{+\infty} e^{-\lambda t}(X^{u,x}_t)^m\,dt 
+ \int_0^{+\infty} e^{-\lambda t}\,dt \biggr]  \notag\\
&\le Cx^m  \int_0^{+\infty} e^{(\lambda_m-\lambda)t}\,dt + \frac{C}{\lambda} \notag\\
\label{eqn:J_estimate}
&= \frac{Cx^m}{\lambda-\lambda_m} + \frac{C}{\lambda} < +\infty \;.
\end{align}
\end{proof}

%--------------------------------------------------------------------------------
%	VALUE FUNCTION
%--------------------------------------------------------------------------------

\section{Properties of the value function}\label{properties}

In this section we explore some properties of the value function.

\begin{prop}
\label{prop:v_properties}
The value function given in \eqref{eqn:infinitetime_pb} satisfies the following properties:
\begin{itemize}
\item $v(x,z)\ge0$ $\forall (x,z)\in(0,+\infty)\times\mathbb{R}$;
%\item $v$ is increasing in $x>0$, i.e. $0<x\le x' \Rightarrow v(x,z)\le v(x',z)$;
\item $\exists M>0$ such that $v(x,z)\le M(1+x^m)$ $\forall (x,z)\in(0,+\infty)\times\mathbb{R}$.
\end{itemize}
If in addition $\exists \tilde{C}>0$ such that $f(x,z,0)\le \tilde{C}x^m$ $\forall (x,z)\in(0,+\infty)\times\mathbb{R}$ for some $m>0$, then 
\begin{itemize}
\item $\exists \tilde{M}>0$ such that $v(x,z)\le \tilde{M}x^m$ $\forall (x,z)\in(0,+\infty)\times\mathbb{R}$;
\item $v(0^+,z)=0$ $\forall z\in\mathbb{R}$.
\end{itemize}
\end{prop}
\begin{proof}
Using Assumption \ref{ass:f}, we easily obtain that $v$ is nonnegative. 
%Now let us take $0<x<x'$. By equation \eqref{eqn:Xsol} we see that $X^{u,x}_t<X^{u,x'}_t$ $\forall t\ge0$ with probability $1$. Using the monotonicity of $f$ we get that
%\[
%\mathbb{E}\biggl[ \int_0^{+\infty} e^{-\lambda t}f(X^{u,x}_t,Z_t,u_t)\,dt \biggr] 
%	\le \mathbb{E}\biggl[ \int_0^{+\infty} e^{-\lambda t}f(X^{u,x'}_t,Z_t,u_t)\,dt \biggr] 
%	\qquad \forall u\in\mathcal{U}\;.
%\]
%Taking the infimum over $\mathcal{U}$ of both sides, we find that the $v$ is increasing. \\
Now manipulating equation \eqref{eqn:J_estimate} we obtain that $\forall u \in \mathcal{U}$, and  $(x,z)\in(0,+\infty)\times\mathbb{R}$

\[
\begin{split}
J(x,z,u) &\le \frac{C}{\lambda-\lambda_m} \biggl( x^m + \frac{\lambda-\lambda_m}{\lambda} \biggr) \\
&\le \frac{C}{\lambda-\lambda_m} (1+ x^m) \;,
\end{split}
\] 
hence $v(x,z)\le M(1+x^m)$ $\forall (x,z)\in(0,+\infty)\times\mathbb{R}$ with $M=\frac{C}{\lambda-\lambda_m}$. 
Now we assume that $f(x,z,0)\le \tilde{C}x^m$ $\forall (x,z)\in(0,+\infty)\times\mathbb{R}$. Then
\[
v(x,z) \le J(x,z,0) \le \tilde{M}x^m \;,
\]
choosing $\tilde{M}=\frac{\tilde C}{\lambda-\lambda_m}$ (by imitation of the proof of Proposition \ref{prop:Jfinite}). This in turn implies that
\[
0 \le v(x,z) \le \tilde{M}x^m \quad  \forall (x,z)\in(0,+\infty)\times\mathbb{R} \quad \Rightarrow \quad v(0^+,z)=0 \quad \forall z\in\mathbb{R} \;.
\]
\end{proof}

\begin{prop}\label{INCR}
Suppose that $f$ is increasing in $x\in(0,+\infty)$. Then $v$ is increasing in $x>0$, i.e. $0<x\le x' \Rightarrow v(x,z)\le v(x',z)$ $\forall z \in \mathbb{R}$.
\end{prop}

\begin{proof}
Let us take $0<x\leq x'$. By equation \eqref{eqn:Xsol} we see that $\forall u\in\mathcal{U}$, $X^{u,x}_t \leq X^{u,x'}_t$ $\forall t\ge0$ $\mathbb{P}-a.s.$. Using the monotonicity of $f$ we get that
\[
\mathbb{E}\biggl[ \int_0^{+\infty} e^{-\lambda t}f(X^{u,x}_t,Z_t,u_t)\,dt \biggr] 
	\leq \mathbb{E}\biggl[ \int_0^{+\infty} e^{-\lambda t}f(X^{u,x'}_t,Z_t,u_t)\,dt \biggr] 
	\qquad \forall u\in\mathcal{U}\;.
\]
Taking the infimum over $\mathcal{U}$ of both sides, we obtain our statement.
\end{proof}

The following proposition is also useful when infinite horizon problems are studied.

\begin{prop}
This result hold true, $\forall u \in \mathcal{U}$, and  $(x,z)\in(0,+\infty)\times\mathbb{R}$

\[
\lim_{T\to+\infty}e^{-\lambda T}\mathbb{E}[v(X^{u,x}_T,z)] = 0 
 \;.
\]
\end{prop}
\begin{proof}
Using Proposition \ref{prop:v_properties} and equation \eqref{eqn:EXm_finite} we find that
\[
\begin{split}
e^{-\lambda T}\mathbb{E}[v(X^{u,x}_T,z)] &\le e^{-\lambda T}M (1+ \mathbb{E}[(X^{u,x}_T)^m ]) \\
&\le e^{-\lambda T}M (1+x^m e^{\lambda_m T}) \;.
\end{split}
\]
Taking $T\to+\infty$, this quantity converges to $0$ since $\lambda>\lambda_m$.
\end{proof}

\begin{prop}
\label{prop:vconvex}
Suppose that $f$ is convex in $x>0$. Moreover, assume that an optimal control $u^* \in\mathcal{U}$ exists for the problem \eqref{eqn:infinitetime_pb} for any $(x,z)\in(0,+\infty)\times\mathbb{R}$. Then the value function in \eqref{eqn:infinitetime_pb} is convex in $x \in (0, + \infty)$.
\end{prop}
\begin{proof}
Let us take $z\in\mathbb{R}$, $x,x'>0$ and $k\in[0,1]$. Defining $x_k=kx+(1-k)x'$, by equation \eqref{eqn:Xsol} we easily get that, $\forall u\in\mathcal{U}$
\[
\begin{split}
X^{u,x_k}_t &= (kx+(1-k)x')\,
e^{\int_0^t (r(Z_s)-g(Z_s,u_s)-u_s)\,ds - \frac{1}{2}\sigma^2t}e^{\sigma W_t} \\
&= kX^{u,x}_t+(1-k)X^{u,x'}_t \qquad \forall t \geq 0, \mathbb{P}-a.s. \;.
\end{split}
\]
Now the convexity of $f$ w.r.t. the first variable $x \in \mathbb R$, implies that,  $\forall u\in\mathcal{U}$
\[
\begin{split}
f(X^{u,x_k}_t,Z_t,u_t) &= f(kX^{u,x}_t+(1-k)X^{u,x'}_t,Z_t,u_t) \\
&\le kf(X^{u,x}_t,Z_t,u_t)+(1-k)f(X^{u,x'}_t,Z_t,u_t) 
\qquad \forall t \geq 0, \mathbb{P}-a.s. \;.
\end{split}
\]
Hence
\[
\begin{split}
J(x_k,z,u) &= \mathbb{E}\biggl[ \int_0^{+\infty} e^{-\lambda t}f(X^{u,x_k}_t,Z_t,u_t)\,dt \biggr] \\
&\le kJ(x,z,u)+(1-k)J(x',z,u) \qquad \forall u\in\mathcal{U} \;.
\end{split}
\]
Taking the infimum over $u \in \mathcal{U}$ in the left side yields
\[
v(x_k,z) \le kJ(x,z,u)+(1-k)J(x',z,u) \qquad \forall u\in\mathcal{U} \;. 
\]
Taking $u=u^*$ we finally get that
\[
v(x_k,z) \le kJ(x,z,u^*)+(1-k)J(x',z,u^*) = kv(x,z) +(1-k)v(x',z) \;.
\]
\end{proof}

\begin{prop}\label{prop:vcont}
Let us assume the following hypotheses:
\begin{itemize}
\item $b_Z$ and $\sigma_Z$ are Lipschitz continuous functions on $z\in\mathbb{R}$;
\item $f$ is continuous in $(x,z)\in(0,+\infty)\times\mathbb{R}$, uniformly in $u\in[-U_1,U_2]$;
\item $r$ is continuous in $z\in\mathbb{R}$;
\item $g$ is continuous in $z\in\mathbb{R}$, uniformly in $u\in[-U_1,U_2]$.
\end{itemize}
Then the value function is continuous in $(x,z)\in(0,+\infty)\times\mathbb{R}$.
\end{prop}
%%%
\begin{proof}
Let us denote by $(X^{u,x,z},Z^z)=\{ (X^{u,x,z}_t, Z^z_t)\}_{t \geq 0}$ the solution to the system of equations \eqref{eqn:Z} and \eqref{eqn:debt} with initial data $(x,z) \in (0,+\infty)\times\mathbb{R}$. By classical results on SDE, the process $Z^z$ depends continuously on the initial data $z \in \mathbb{R}$, moreover by  \eqref{eqn:Xsol} we have that for any $u \in \mathcal{U}$
$$X^{u,x,z}_t = xe^{\int_0^t (r(Z^z_s)-g(Z^z_s,u_s)-u_s)\,ds - \frac{1}{2}\sigma^2t}e^{\sigma W_t} \qquad \forall t \geq 0, \mathbb{P}-a.s. \;.
$$
Let $\{(x_n, z_n)\}_{n\geq 0}$ be any sequence in $(0,+\infty)\times\mathbb{R}$ converging to  $(x,z) \in (0,+\infty)\times\mathbb{R}$ as $n \to + \infty$, then  $Z^{z_n}_t \to Z^{z}_t$ $\forall t \geq 0$, as  $n \to + \infty$ and by the dominated convergence theorem we have that $\forall t \geq 0$
$$X^{u,x_n,z_n}_t \to X^{u,x,z}_t \qquad \text{uniformly on} \; u\in\mathcal{U} \;, \text{as} \;  n \to + \infty$$
and 
$$J(x_n,z_n,u) \to J(x,z,u) \qquad \text{uniformly on} \; u\in\mathcal{U} \;, \text{as} \; n \to + \infty \;,$$
which finally implies continuity of $v(x,z) = \inf_{u\in\mathcal{U}} J(x,z,u)$ in $(x,z)\in(0,+\infty)\times\mathbb{R}$.
\end{proof}

\section{Characterization of the value function}\label{characterization}

In this section we aim to characterize the value function $v$ given in equation \eqref{eqn:infinitetime_pb}. Precisely, we prove that it is a viscosity solution of the Hamilton-Jacobi-Bellman equation associated to our problem (see equation \eqref{eqn:HJB_general} below). To obtain this result only the continuity of $v$ is required. If, in addition, the HJB equation admits a classical solution, then it turns out to be the value function. In this case $v$ will satisfy some additional regularity conditions.\\

We begin finding the HJB equation associated to our problem.
\begin{remark}
For any $(x,z)\in(0,+\infty)\times\mathbb{R}$ and any $u\in [-U_1,U_2]$ the Markov generator of $(X^{u}, Z)$ is given by the following differential operator%
\footnote{This is  a simple application of It\^o's formula.}:
\begin{equation}
\label{eqn:generator}
\mathcal{L}^u \phi(x,z) = x[r(z)-g(z,u)-u]\frac{\partial{\phi}}{\partial{x}}(x,z) 
+ \frac{1}{2}\sigma^2x^2\frac{\partial^2{\phi}}{\partial{x^2}}(x,z) 
+\rho\sigma x\sigma_Z(z)\frac{\partial^2{\phi}}{\partial{x}\partial{z}}(x,z) +\mathcal{L}^Z\phi(x,z)\;,
\end{equation}
where $\phi\colon(0,+\infty)\times\mathbb{R}\to\mathbb{R}$ is a function on $\mathcal{C}^{2,2}((0,+\infty)\times\mathbb{R})$ and $\mathcal{L}^Z$ denotes the operator
\[
\mathcal{L}^Z\phi (x,z) = b_Z(z)\frac{\partial{\phi}}{\partial{z}}(x,z) 
+ \frac{1}{2}\sigma_Z(z)^2\frac{\partial^2{\phi}}{\partial{z^2}}(x,z)  \;.
\]
\end{remark}

The value function in equation \eqref{eqn:infinitetime_pb}, if sufficiently regular, is expected to solve the HJB equation, which is given by
\begin{equation}
\label{eqn:HJB_general}
\inf_{u\in[-U_1,U_2]}\{ \mathcal{L}^u v(x,z) + f(x,z,u) - \lambda v(x,z) \} =0 \;.
\end{equation}

Before stating the main result, we briefly recall the definition of viscosity solution to equation \eqref{eqn:HJB_general}. Let us notice that, in general, one would require that a function $w$ is locally bounded in order to be the solution of a PDE in viscosity sense (see for instance \cite[Chapter 4]{pham:stochastic_control}). However, we already know that $v$ given in equation \eqref{eqn:infinitetime_pb} is continuous, hence we can directly refer to the special case of continuous functions.

\begin{definition}
Let $w\colon(0,+\infty)\times\mathbb{R}\to[0,+\infty)$ be continuous. We say that 
\begin{itemize}
\item $w$ is a viscosity subsolution of equation \eqref{eqn:HJB_general} if
\begin{equation}
\label{eqn:vsubsol}
\inf_{u\in[-U_1,U_2]}\{ \mathcal{L}^u \varphi(\bar{x},\bar{z}) 
+ f(\bar{x},\bar{z},u) - \lambda \varphi(\bar{x},\bar{z}) \} \ge 0 \;,
\end{equation}
for all $(\bar{x},\bar{z})\in(0,+\infty)\times\mathbb{R}$ and for all $\varphi\in\mathcal{C}^{2,2}((0,+\infty)\times\mathbb{R})$ such that $(\bar{x},\bar{z})$ is a maximum point of $w-\varphi$;
% infatti qui $v\le\phi$
\item $w$ is a viscosity supersolution of equation \eqref{eqn:HJB_general} if
\begin{equation}
\label{eqn:vsupersol}
\inf_{u\in[-U_1,U_2]}\{ \mathcal{L}^u \varphi(\bar{x},\bar{z}) 
+ f(\bar{x},\bar{z},u) - \lambda \varphi(\bar{x},\bar{z}) \} \le 0 \;,
\end{equation}
for all $(\bar{x},\bar{z})\in(0,+\infty)\times\mathbb{R}$ and for all $\varphi\in\mathcal{C}^{2,2}((0,+\infty)\times\mathbb{R})$ such that $(\bar{x},\bar{z})$ is a minimum point of $w-\varphi$;
\item $w$ is a viscosity solution of equation \eqref{eqn:HJB_general} if it is a viscosity subsolution and a supersolution.
\end{itemize}
\end{definition}

\begin{theorem}\label{thm:valuefun_viscosity}
Under the hypotheses of Proposition \ref{prop:vcont}, the value function $v$ given in equation \eqref{eqn:infinitetime_pb} is a viscosity solution of the HJB equation \eqref{eqn:HJB_general}.
\end{theorem}
\begin{proof}
We adapt the proof of \cite[Chapter 4.3]{pham:stochastic_control} to our framework.
Proposition \ref{prop:vcont} ensures the continuity of $v$. 
Let $(\bar{x},\bar{z})\in(0,+\infty)\times\mathbb{R}$ and take a test function $\varphi\in\mathcal{C}^{2,2}((0,+\infty)\times\mathbb{R})$ such that
\[
0 = (v-\varphi)(\bar{x},\bar{z}) = \max_{(x,z)\in(0,+\infty)\times\mathbb{R}}(v-\varphi)(x,z) \;.
\]
Let us observe that $v\le\varphi$ by construction. Since $v$ is continuous, there exists a sequence $\{(x_n,z_n)\}_{n\ge1}$ such that
\[
(x_n,z_n) \to (\bar{x},\bar{z}) \qquad \text{and } \qquad v(x_n,z_n) \to v(\bar{x},\bar{z}) 
\qquad \text{as } n\to+\infty \;.
\]
Correspondingly, we must have that
\[
\gamma_n \doteq v(x_n,z_n) - \varphi(x_n,z_n) \to 0
\qquad \text{as } n\to+\infty \;.
\]
Now let us consider a control $\bar{u}_t=\bar{u}$ $\forall t>0$, for some arbitrary constant $\bar{u}\in[-U_1,U_2]$. Moreover, let introduce a sequence of stopping times $\{\tau_n\}_{n\ge0}$ as follows:
\[
\tau_n = \inf\{ s\ge0 \mid \max\{\abs{X^{\bar{u},x_n}_s-x_n},\abs{Z^{z_n}_s-z_n}\} > \epsilon \} \land h_n 
\qquad n\ge1\;,
\]
for some fixed $\epsilon>0$ and $\{h_n\}_{n\ge1}$ such that
\[
h_n\to0 \;, \qquad \frac{\gamma_n}{h_n} \to 0 \qquad \text{as } n\to+\infty \;.
\]
Here $\{Z^{z_n}_t\}_{t\ge0}$ denotes the solution of the SDE \eqref{eqn:Z} with initial condition $Z^{z_n}_0=z_n$. By the dynamic programming principle (see e.g. \cite[Theorem 3.3.1]{pham:stochastic_control}) for any $n\ge1$ we have that
\[
v(x_n,z_n) \le \mathbb{E}\biggl[ \int_0^{\tau_n} e^{-\lambda t}
f(X^{\bar{u},x_n}_t,Z^{z_n}_t,\bar{u})\,dt + e^{-\lambda \tau_n}v(X^{\bar{u},x_n}_{\tau_n},Z^{z_n}_{\tau_n}) \biggr] \;,
\]
hence
\[
\varphi(x_n,z_n) + \gamma_n \le \mathbb{E}\biggl[ \int_0^{\tau_n} e^{-\lambda t}
f(X^{\bar{u},x_n}_t,Z^{z_n}_t,\bar{u})\,dt 
+ e^{-\lambda \tau_n}\varphi(X^{\bar{u},x_n}_{\tau_n},Z^{z_n}_{\tau_n}) \biggr] \;.
\]
Applying It\^o's formula we get that
\begin{equation}
\label{eqn:phi_viscosity}
e^{-\lambda \tau_n}\varphi(X^{\bar{u},x_n}_{\tau_n},Z^{z_n}_{\tau_n}) = \varphi(x_n,z_n) + 
\int_0^{\tau_n} e^{-\lambda t}[\mathcal{L}^{\bar{u}} \varphi(X^{\bar{u},x_n}_t,Z^{z_n}_t) 
- \lambda\varphi(X^{\bar{u},x_n}_t,Z^{z_n}_t)] \,dt + M_{\tau_n} \;,
\end{equation}
where
\[
M_t = \int_0^t e^{-\lambda s}\frac{\partial{\varphi}}{\partial{x}}(X^{\bar{u},x_n}_s,Z^{z_n}_s)\sigma X^{\bar{u},x_n}_s \,dW_s
+ \int_0^t e^{-\lambda s}\frac{\partial{\varphi}}{\partial{z}}(X^{\bar{u},x_n}_s,Z^{z_n}_s)\sigma_Z(Z^{z_n}_s) \,dW^Z_s \;.
\]
Clearly $\{M_t\}_{t\ge0}$ is a local martingale (having  $\{\tau_n\}_{n\ge0}$ as localizing sequence of stopping times) because the integrand functions are continuous and hence bounded on the compact sets. 
% (notice that the derivatives of $\phi$ are continuous and hence bounded in closed domains.
Taking expectations in equation \eqref{eqn:phi_viscosity} and using the previous inequality yields
\[
\begin{split}
&\gamma_n+\varphi(x_n,z_n) \\
&= \gamma_n - \mathbb{E}\biggl[ \int_0^{\tau_n} e^{-\lambda t}[\mathcal{L}^{\bar{u}} \varphi(X^{\bar{u},x_n}_t,Z^{z_n}_t) - \lambda\varphi(X^{\bar{u},x_n}_t,Z^{z_n}_t)]\,dt \biggr]
+ \mathbb{E}\biggl[ e^{-\lambda \tau_n}\varphi(X^{\bar{u},x_n}_{\tau_n},Z^{z_n}_{\tau_n}) \biggr] \\
&\le \mathbb{E}\biggl[ \int_0^{\tau_n} e^{-\lambda t}f(X^{\bar{u},x_n}_t,Z^{z_n}_t,\bar{u})\,dt \biggr]
+ \mathbb{E}\biggl[ e^{-\lambda \tau_n}\varphi(X^{\bar{u},x_n}_{\tau_n},Z^{z_n}_{\tau_n}) \biggr] \;,
\end{split}
\]
that is, dividing by $h_n$ (using that $\tau_n\le h_n$),
\[
\frac{\gamma_n}{h_n} \le \mathbb{E}\biggl[ \frac{1}{h_n}\int_0^{\tau_n}e^{-\lambda t}[\mathcal{L}^{\bar{u}} \varphi(X^{\bar{u},x_n}_t,Z^{z_n}_t) + f(X^{\bar{u},x_n}_t,Z^{z_n}_t,\bar{u}) 
- \lambda\varphi(X^{\bar{u},x_n}_t,Z^{z_n}_t)]\,dt \biggr] \;.
\]
Letting $n\to+\infty$ we have that $X^{\bar{u},x_n}_t\to X^{\bar{u},\bar{x}}_t$ and $Z^{z_n}_t\to Z^{\bar{z}}_t$, $\forall t \geq 0$ $ \mathbb{P}-a.s.$ 
and the right-hand side converges to $\mathcal{L}^{\bar{u}} \varphi(\bar{x},\bar{z}) + f(\bar{x},\bar{z},\bar{u})-\lambda\varphi(\bar{x},\bar{z})$ by the mean value theorem for integrals. Hence
\[
\mathcal{L}^{\bar{u}} \varphi(\bar{x},\bar{z}) + f(\bar{x},\bar{z},\bar{u}) - \lambda \varphi(\bar{x},\bar{z}) \ge 0 \;.
\]
Since $\bar{u}$ is arbitrary, taking the infimum we obtain that $v$ is a viscosity subsolution of equation \eqref{eqn:HJB_general} (see equation \eqref{eqn:vsubsol}).\\

Now we prove that $v$ is a viscosity supersolution. To this end, we take a test function $\varphi\in\mathcal{C}^{2,2}((0,+\infty)\times\mathbb{R})$ such that
\[
0 = (v-\varphi)(\bar{x},\bar{z}) = \min_{(x,z)\in(0,+\infty)\times\mathbb{R}}(v-\varphi)(x,z) \;.
\]
By definition of the value function, we can find a strategy $\{\hat{u}_t\}_{t\ge0}\in\mathcal{U}$ such that
\[
v(x_n,z_n) + h_n^2 \ge \mathbb{E}\biggl[ \int_0^{\tau_n} e^{-\lambda t}
f(X^{\hat{u},x_n}_t,Z^{z_n}_t,\hat{u}_t)\,dt + e^{-\lambda \tau_n}v(X^{\hat{u},x_n}_{\tau_n},Z^{z_n}_{\tau_n}) \biggr] \;,
\]
and hence
\[
\varphi(x_n,z_n) + \gamma_n + h_n^2 \ge \mathbb{E}\biggl[ \int_0^{\tau_n} e^{-\lambda t}
f(X^{\hat{u},x_n}_t,Z^{z_n}_t,\hat{u}_t)\,dt + e^{-\lambda \tau_n}v(X^{\hat{u},x_n}_{\tau_n},Z^{z_n}_{\tau_n}) \biggr] \;.
\]
Using equation \eqref{eqn:phi_viscosity} we obtain that
\begin{multline*}
\gamma_n + h_n^2
+ \mathbb{E}\biggl[ e^{-\lambda \tau_n}\varphi(X^{\hat{u},x_n}_{\tau_n},Z^{z_n}_{\tau_n}) \biggr]
- \mathbb{E}\biggl[ \int_0^{\tau_n} e^{-\lambda t}[\mathcal{L}^{\hat{u}} \varphi(X^{\hat{u},x_n}_t,Z^{z_n}_t) - \lambda \varphi(X^{\hat{u},x_n}_t,Z^{z_n}_t)]\,dt \biggr] \\
\ge \mathbb{E}\biggl[ \int_0^{\tau_n} e^{-\lambda t}f(X^{\hat{u},x_n}_t,Z^{z_n}_t,\hat{u}_t)\,dt \biggr]
+ \mathbb{E}\biggl[ e^{-\lambda \tau_n}\varphi(X^{\hat{u},x_n}_{\tau_n},Z^{z_n}_{\tau_n}) \biggr] \;,
\end{multline*}
and dividing by $h_n$ we get
\[
\begin{split}
\frac{\gamma_n}{h_n} + h_n &\ge \mathbb{E}\biggl[ \frac{1}{h_n}\int_0^{\tau_n}e^{-\lambda t}[\mathcal{L}^{\hat{u}} \varphi(X^{\hat{u},x_n}_t,Z^{z_n}_t) + f(X^{\hat{u},x_n}_t,Z^{z_n}_t,\hat{u}_t) - \lambda \varphi(X^{\hat{u},x_n}_t,Z^{z_n}_t)]\,dt \biggr] \\
&\ge \mathbb{E}\biggl[ \frac{1}{h_n}\int_0^{\tau_n}\inf_{u\in[-U_1,U_2]}\{\mathcal{L}^{u} \varphi(X^{u,x_n}_t,Z^{z_n}_t) + f(X^{u,x_n}_t,Z^{z_n}_t,u) - \lambda \varphi(X^{u,x_n}_t,Z^{z_n}_t) \}\,dt \biggr] \;.
\end{split}
\]
Observing that 
\[
\lim_{n\to+\infty}\frac{\tau_n}{h_n} = \min\{ \lim_{n\to+\infty}\frac{\inf\{ s\ge0 \mid \max\{\abs{X^{\bar{u},x_n}_s-x_n},\abs{Z^{z_n}_s-z_n}\}\ge \epsilon \}}{h_n}, 1 \} = 1 \;,
\]
using the mean value theorem for integrals again, we finally get the inequality
\[
\inf_{u\in[-U_1,U_2]}\{ \mathcal{L}^{u} \varphi(\bar{x},\bar{z}) + f(\bar{x},\bar{z},u) - \lambda\varphi(\bar{x},\bar{z}) \} \le 0 \;,
\]
and hence $v$ is a viscosity supersolution of equation \eqref{eqn:HJB_general} (see equation \eqref{eqn:vsupersol}).
\end{proof}

Now we provide a Verification Theorem based on classical solutions to the HJB equation \eqref{eqn:HJB_general}.

\begin{theorem}\label{verification}
Let $w\colon(0,+\infty)\times\mathbb{R}\to[0,+\infty)$ be a function in $\mathcal{C}^{2,2}((0,+\infty)\times\mathbb{R})$ and suppose that there exists a constant $C_1>0$ such that
\[
\abs{w(x,z)} \le C_1(1+\abs{x}^m) \qquad \forall (x,z)\in(0,+\infty)\times\mathbb{R} \;.
\]
%for some $\gamma>0$. 
If 
\begin{equation}
\label{eqn:HJBverification1}
\inf_{u\in[-U_1,U_2]}\{ \mathcal{L}^u w(x,z) + f(x,z,u) - \lambda w(x,z) \}\ge 0 
\qquad \forall (x,z)\in(0,+\infty)\times\mathbb{R} \;,
\end{equation}
%and
%\begin{equation}
%\label{eqn:verif1}
%\limsup_{T\to+\infty}e^{-\lambda T} \mathbb{E}[w(X^{u,x}_T)] \le 0 
%\qquad \forall (x,z)\in(0,+\infty)\times\mathbb{R}), \forall u\in\mathcal{U}\;,
%\end{equation}
then $w(x,z)\le v(x,z)$ $\forall (x,z)\in(0,+\infty)\times\mathbb{R}$.\\

Now suppose that there exists a $[-U_1,U_2]$-valued measurable function $u^*(x,z)$ such that
\begin{multline}
\label{eqn:HJBverification2}
\inf_{u\in[-U_1,U_2]}\{ \mathcal{L}^u w(x,z) + f(x,z,u) - \lambda w(x,z) \} \\
= \mathcal{L}^{u^*} w(x,z) + f(x,z,u^*(x,z)) - \lambda w(x,z) = 0
\quad \forall (x,z)\in(0,+\infty)\times\mathbb{R} \;.
\end{multline}
Then $w(x,z) = v(x,z)$ $\forall (x,z)\in(0,+\infty)\times\mathbb{R}$ and $u^*=\{u^*(X^{u^*,x}_t, Z_t)\}_{t\ge0}$ is an optimal (Markovian) control.
\end{theorem}

\begin{proof}
Let $w\in\mathcal{C}^{2,2}((0,+\infty)\times\mathbb{R})$. Let us introduce a sequence of stopping times defined by
\begin{multline}
\label{eqn:tau_n}
\tau_n \doteq \inf\{t\ge0 \mid \int_0^t e^{-2\lambda s}\abs*{\frac{\partial{w}}{\partial{x}}(X^{u,x}_s,Z_s)\sigma X^{u,x}_s}^2 \,ds \ge n\} \\
\land \inf\{t\ge0 \mid \int_0^t e^{-2\lambda s}\abs*{\frac{\partial{w}}{\partial{z}}(X^{u,x}_s,Z_s)\sigma_Z(Z_s) }^2 \,ds \ge n\} \;, \qquad n=1,2,\dots \;.
\end{multline}
Applying It\^o's formula to $e^{-\lambda t}w(X^{u,x}_t,Z_t)$ for any arbitrary $u\in\mathcal{U}$ we get
\begin{multline*}
e^{-\lambda T\land\tau_n}w(X^{u,x}_{T\land\tau_n},Z_{T\land\tau_n}) = w(x,z)
+ \int_0^{T\land\tau_n} e^{-\lambda t}[\mathcal{L}^u w(X^{u,x}_t,Z_t) - \lambda w(X^{u,x}_t,Z_t)]\,dt \\
+ \int_0^{T\land\tau_n} e^{-\lambda t}\frac{\partial{w}}{\partial{x}}(X^{u,x}_t,Z_t)\sigma X^{u,x}_t \,dW_t
+ \int_0^{T\land\tau_n} e^{-\lambda t}\frac{\partial{w}}{\partial{z}}(X^{u,x}_t,Z_t)\sigma_Z(Z_t) \,dW^Z_t \;.
\end{multline*}
The last integrals are real martingales by definition of $\tau_n$ (see equation \eqref{eqn:tau_n}), hence taking the expectation and using the inequality \eqref{eqn:HJBverification1} gives
\[
\mathbb{E}[e^{-\lambda T\land\tau_n}w(X^{u,x}_{T\land\tau_n},Z_{T\land\tau_n})] \ge w(x,z) 
- \mathbb{E}\biggl[ \int_0^{T\land\tau_n} e^{-\lambda t}f(X^{u,x}_t,Z_t,u_t)\,dt \biggr] \;.
\]
By the growth condition on $w$ and Proposition \ref{prop:Jfinite} we can apply the dominated convergence theorem, so that letting $n\to+\infty$ gives%
\footnote{Notice that $\tau_n\to+\infty$ because the integrand functions in equation \eqref{eqn:tau_n} are continuous by our assumptions.}
\[
\mathbb{E}[e^{-\lambda T}w(X^{u,x}_T,Z_T)] \ge w(x,z) 
- \mathbb{E}\biggl[ \int_0^T e^{-\lambda t}f(X^{u,x}_t,Z_t,u_t)\,dt \biggr]
 \qquad \forall u\in\mathcal{U} \;.
\]
Recalling equation \eqref{eqn:EXm_finite}, since
\[
\begin{split}
\limsup_{T\to+\infty}e^{-\lambda T} \mathbb{E}[w(X^{u,x}_T,Z_T)]
&\le C_1\limsup_{T\to+\infty}e^{-\lambda T} \mathbb{E}[1+(X^{u,x}_T)^m] \\
&\le 0 \qquad \forall (x,z)\in(0,+\infty)\times\mathbb{R}, \forall u\in\mathcal{U}\;,
\end{split}
\]
we can send $T\to+\infty$ to obtain that
\[
w(x,z)  \le \mathbb{E}\biggl[ \int_0^{+\infty} e^{-\lambda t}f(X^{u,x}_t,Z_t,u_t)\,dt \biggr]
\qquad \forall u\in\mathcal{U} \;,
\]
which implies the first inequality $w(x,z)\le v(x,z)$ $\forall (x,z)\in(0,+\infty)\times\mathbb{R}$.\\

Now we can repeat the same argument choosing the control $\{u^*_t=u^*(X^{u^*,x}_t,Z_t)\}_{t\ge0}$, obtaining the equality
\[
\mathbb{E}[e^{-\lambda T}w(X^{u^*,x}_T,Z_T)] = w(x,z) 
- \mathbb{E}\biggl[ \int_0^T e^{-\lambda t}f(X^{u^*,x}_t,Z_t,u^*_t)\,dt \biggr] \;.
\]
Using the fact that
\[
\liminf_{T\to+\infty}e^{-\lambda T} \mathbb{E}[w(X^{u^*,x}_T,Z_T)] \ge 0 
\qquad \forall (x,z)\in(0,+\infty)\times\mathbb{R} \;,
\]
sending $T\to+\infty$ we deduce
\[
w(x,z) \ge \mathbb{E}\biggl[ \int_0^{+\infty} e^{-\lambda t}f(X^{u^*,x}_t,Z_t,u^*(X^{u^*,x}_t,Z_t))\,dt \biggr] \;,
\]
which shows that $w(x,z) = v(x,z)$ $\forall (x,z)\in(0,+\infty)\times\mathbb{R}$ and that $u^*$ is an optimal control.
\end{proof}

%--------------------------------------------------------------------------------
%	OPTIMAL STRATEGY
%--------------------------------------------------------------------------------

\section{The optimal fiscal policy}
\label{optimal_stra}

In view of the Verification Theorem in this section we aim at investigating the optimal candidate fiscal policy and its properties. 
Putting the expression \eqref{eqn:generator} into the HJB equation \eqref{eqn:HJB_general} and gathering the terms which depend on $u$ gives
\begin{equation}
\label{eqn:HJB_extended}
\begin{split}
&\inf_{u\in[-U_1,U_2]}\{x[-g(z,u)-u]\frac{\partial{v}}{\partial{x}}(x,z) + f(x,z,u) \} +
x r(z)\frac{\partial{v}}{\partial{x}}(x,z)  \\
&+ \frac{1}{2}\sigma^2x^2\frac{\partial^2{v}}{\partial{x^2}}(x,z) 
+\rho\sigma x\sigma_Z(z)\frac{\partial^2{v}}{\partial{x}\partial{z}}(x,z) 
+\mathcal{L}^Zv(x,z) - \lambda v(x,z) = 0 \;.
\end{split}
\end{equation}

We look for a minimizer of the following function:
\begin{equation}
\label{eqn:H}
H(x,z,u) = -x[g(z,u) + u]\frac{\partial{v}}{\partial{x}}(x,z) + f(x,z,u) \;.
\end{equation}
%with\begin{equation}\label{eqn:dH}\frac{\partial{H}}{\partial{u}}(x,z,u) =\frac{\partial{f}}{\partial{u}}(x,z,u) 
%- x\frac{\partial{v}}{\partial{x}}(x,z)\biggl[ \frac{\partial{g}}{\partial{u}}(z,u) + 1 \biggr] \;.
%\end{equation}
 
Our candidate optimal strategy is given through a function $u^*\colon(0,+\infty)\times\mathbb{R}\to[-U_1,U_2]$ which solves the following minimization problem:
\begin{equation}
\label{eqn:min_pb}
H(x,z,u^*(x,z)) = \min_{u\in[-U_1,U_2]}H(x,z,u) \;.
\end{equation}

To this end, we make use of the following assumptions.

\begin{assumption}
We assume that
\begin{itemize}
\item $g$ is continuous and differentiable in $u\in[-U_1,U_2]$;
\item $f$ is continuous and differentiable in $u\in[-U_1,U_2]$;
\item $v \in C^{2,2}((0, + \infty) \times \mathbb R)$ is a classical solution of equation \eqref{eqn:HJB_extended}.
\end{itemize}
\end{assumption}

We first state the existence and uniqueness of the minimizer of problem \eqref{eqn:min_pb}, then we give a characterization of it.

\begin{prop}
\label{prop:Hconvex}
The problem \eqref{eqn:min_pb} admits a minimizer $u^*(x,z)\in[-U_1,U_2]$ for any $(x,z)\in(0,+\infty)\times\mathbb{R}$. Moreover, if $H$ is strictly convex, the minimizer is also unique.
\end{prop}
\begin{proof}
The existence of a minimizer immediately follows by the compactness of the interval $[-U_1,U_2]$ and the Weierstrass Theorem. Moreover, the strictly convexity of $H$ implies the uniqueness of the minimizer by classical arguments.  
\end{proof}
 
 \begin{remark}
 %$f$ convex and $g$ concave in $u\in[-U_1,U_2]$, at least one of them in strict sense, implies $H$ is strictly convex.
 %$f$ strictly convex and $g$ linear on $u\in[-U_1,U_2]$, implies $H$ is strictly convex. NO! dipende dal segno di $frac{\partial{v}}{\partial{x}}$
 If $f$ and $g$ are both linear on $u\in[-U_1,U_2]$, $H(x,z,u)$ is linear and the minimizer is also unique.
  \end{remark}

%Let us denote by $\hat{u}(x,z)$ the solution to
%\begin{equation}
%\label{eqn:firstordercond}
%\frac{\partial{H}}{\partial{u}}(x,z,u) = 0
%\quad \Leftrightarrow \quad
%x\frac{\partial{v}}{\partial{x}}(x,z)\biggl[ \frac{\partial{g}}{\partial{u}}(z,u) + 1 \biggr]
%= \frac{\partial{f}}{\partial{u}}(x,z,u) \;.
%\end{equation}

\begin{prop}
\label{prop:optimal_Hconvex}
Suppose that  $\frac{\partial^2{H}}{\partial{u}^2}(x,z,u)>0$, $\forall (x,z,u) (0,+\infty)\times\mathbb{R}\in[-U_1,U_2]$. Then there exists a unique minimizer $u^*(x,z)$ for the problem \eqref{eqn:min_pb}. Moreover, $u^*(x,z)$ admits the following expression:
\begin{multline*}
u^*(x,z) = \\
\begin{cases}
	-U_1
	& (x,z)\in\Set{(x,z)\in(0,+\infty)\times\mathbb{R} \mid\frac{\partial{f}}{\partial{u}}(x,z,-U_1) 
	\ge x\frac{\partial{v}}{\partial{x}}(x,z)\bigl[ \frac{\partial{g}}{\partial{u}}(z,-U_1) + 1 \bigr] }
	\\
	U_2 
	& (x,z)\in\Set{(x,z)\in(0,+\infty)\times\mathbb{R} \mid\frac{\partial{f}}{\partial{u}}(x,z,U_2) 
	\le x\frac{\partial{v}}{\partial{x}}(x,z)\bigl[ \frac{\partial{g}}{\partial{u}}(z,U_2) + 1 \bigr] }
	\\
	\hat{u}(x,z) & \text{otherwise} \;,
\end{cases}
\end{multline*}
where $\hat{u}(x,z)$ denotes the solution to
\begin{equation}
\label{eqn:firstordercond}
%\frac{\partial{H}}{\partial{u}}(x,z,u) = 0
%\quad \Leftrightarrow \quad
x\frac{\partial{v}}{\partial{x}}(x,z)\biggl[ \frac{\partial{g}}{\partial{u}}(z,u) + 1 \biggr]
= \frac{\partial{f}}{\partial{u}}(x,z,u) \quad \forall (x,z) \in (0,+\infty)\times\mathbb{R}\;.
\end{equation}
\end{prop}
\begin{proof}
Existence and uniqueness of the minimizer are guaranteed by classical arguments, since $H$ is continuous and strictly convex in $u\in[-U_1,U_2]$. %Now we have only three cases. %Let us take $(x,z)\in(0,+\infty)\times\mathbb{R}$. 
Observing that
\begin{equation}
\label{eqn:dH}
\frac{\partial{H}}{\partial{u}}(x,z,u) =
\frac{\partial{f}}{\partial{u}}(x,z,u) 
 - x\frac{\partial{v}}{\partial{x}}(x,z)\biggl[ \frac{\partial{g}}{\partial{u}}(z,u) + 1 \biggr] \;,
\end{equation}
we have only three cases.
\begin{enumerate}
\item If $(x,z) \in (0,+\infty)\times\mathbb{R}$ are such that $\frac{\partial{H}}{\partial{u}}(x,z,-U_1)\ge0$, i.e.
\[
\frac{\partial{f}}{\partial{u}}(x,z,-U_1) 
\ge x\frac{\partial{v}}{\partial{x}}(x,z)\bigl[ \frac{\partial{g}}{\partial{u}}(z,-U_1) + 1 \bigr] \;,
\]
then we must have 
\[
\frac{\partial{H}}{\partial{u}}(x,z,u)\geq \frac{\partial{H}}{\partial{u}}(x,z,-U_1) 
\ge0 \qquad \forall u\in[-U_1,U_2] \;,
\]
because of the convexity of $H$. Hence $H$ is increasing on the whole interval $[-U_1,U_2]$ and the minimizer turns out to be $u^*(x,z) = -U_1$.

\item If $(x,z) \in (0,+\infty)\times\mathbb{R}$ are such that $\frac{\partial{H}}{\partial{u}}(x,z,U_2)\le0$, i.e. 
\[
\frac{\partial{f}}{\partial{u}}(x,z,U_2) 
\le x\frac{\partial{v}}{\partial{x}}(x,z)\bigl[ \frac{\partial{g}}{\partial{u}}(z,U_2) + 1 \bigr] \;,
\]
then we have that
\[
\frac{\partial{H}}{\partial{u}}(x,z,u) \leq \frac{\partial{H}}{\partial{u}}(x,z,U_2)\le 0
\qquad \forall u\in[-U_1,U_2] \;,
\]
so that $H$ is decreasing in $u\in[-U_1,U_2]$ and therefore $u^*(x,z) = U_2$ is the minimizer.\\

\item Finally, when $\frac{\partial{H}}{\partial{u}}(x,z,-U_1)<0$ and $\frac{\partial{H}}{\partial{u}}(x,z,U_2)>0$, since  $\frac{\partial{H}}{\partial{u}}$ is continuous in $u\in[-U_1,U_2]$, there exists $\hat{u}(x,z)$ such that $\frac{\partial{H}}{\partial{u}}(x,z,\hat{u}(x,z))=0$ (see equation \eqref{eqn:firstordercond}) and this stationary point coincides with the minimizer.
\end{enumerate}
\end{proof}

%\ourcomment{$\frac{\partial{H}}{\partial{u}}(x,z,\hat{u}(x,z))$ continua?}

\subsection{Some special cases}
In this section we investigate some cases of interest. 
Let us first establish a slightly general result, when the cost function does not depend explicitly on $u$.

\begin{prop}\label{prop1_f(x,z)}
Suppose that the cost function $f$ does not depend on $u$, i.e. $f(x,z,u) = f(x,z)$ and $\frac{\partial^2{H}}{\partial{u}^2}(x,z,u)>0$, $\forall (x,z,u) (0,+\infty)\times\mathbb{R}\in[-U_1,U_2]$.
Then the minimizer of  \eqref{eqn:min_pb} is given by
\begin{equation}
\label{eqn:u*z}
u^*(x,z) = 
\begin{cases} -U_1
& (x,z)\in\Set{(x,z)\in(0,+\infty)\times\mathbb{R} \mid  0
	\ge \frac{\partial{v}}{\partial{x}}(x,z)\bigl[ \frac{\partial{g}}{\partial{u}}(z,-U_1) + 1 \bigr] }
	\\
	U_2 
	& (x,z)\in\Set{(x,z)\in(0,+\infty)\times\mathbb{R}  \mid  0 \le \frac{\partial{v}}{\partial{x}}(x,z)\bigl[ \frac{\partial{g}}{\partial{u}}(z,U_2) + 1 \bigr] }
	\\
	\hat{u}(z) & \text{otherwise} \;,
\end{cases}
\end{equation}
where $\hat{u}(z)$ is the unique solution to
\begin{equation}\label{eqn:hatu_match}
\frac{\partial{g}}{\partial{u}}(z,u) = -1  \quad \forall z \in \mathbb{R} \;.
\end{equation}
\end{prop}
\begin{proof}
Observing that
\begin{equation}\label{C1}
\frac{\partial{H}}{\partial{u}}(x,z,u) =
 - x\frac{\partial{v}}{\partial{x}}(x,z)\biggl[ \frac{\partial{g}}{\partial{u}}(z,u) + 1 \biggr] \;,
\end{equation}
\begin{equation}\label{C2}
\frac{\partial^2 {H}}{\partial{u}^2}(x,z,u) =
 - x\frac{\partial{v}}{\partial{x}}(x,z)
 \frac{\partial^2{g}}{\partial{u}^2}(z,u) >0 \;, 
\end{equation}
hence $\frac{\partial{v}}{\partial{x}}(x,z) \neq 0$, $\forall (x,z) \in (0,+\infty)\times\mathbb{R}$ and the proof follows the same lines of Proposition \ref{prop:optimal_Hconvex}. 
\end{proof}
%\begin{prop}
%\label{prop:f(x,z)}
%Suppose that the cost function $f$ does not depend on $u$, i.e. $f(x,z,u) = f(x,z)$, and that $f(x,z)$ is increasing  in $x \in (0, + \infty)$. 
%Then the minimizer of equation \eqref{eqn:min_pb} is given by
%\[
%u^*(z) = \argmax_{u\in[-U_1,U_2]}\{ g(z,u) + u \} \;.
%\]
%Moreover, if $g$ is strictly concave in $u\in[-U_1,U_2]$, and $\frac{\partial{g}}{\partial{u}}$ is continuous in $u\in[-U_1,U_2]$, then $u^*(z)$ is given by
%\[
%u^*(z) = -U_1 \lor \hat{u}(z) \land U_2 \;,
%\]
%where $\hat{u}(z)$ is the unique solution to
%\[
%\frac{\partial{g}}{\partial{u}}(z,u) = -1 \;.
%\]
%Precisely, $u^*(z)$ admits the following structure:
%\begin{equation}
%\label{eqn:u*z}
%u^*(z) = 
%\begin{cases}
%	-U_1
%	& z\in\Set{z\in\mathbb{R}\mid\frac{\partial{g}}{\partial{u}}(z,-U_1) < -1}
%	\\
%	U_2 
%	& z\in\Set{z\in\mathbb{R}\mid\frac{\partial{g}}{\partial{u}}(z,U_2) > -1}
%	\\
%	\hat{u}(z) & \text{otherwise} \;.
%\end{cases}
%\end{equation}
%\end{prop}
%\begin{proof} By Proposition \ref{INCR} we have that  $v(x,z)$ is increasing in $x \in (0, + \infty)$ for any $z\in \mathbb R$, hence 
%$\frac{\partial{v}}{\partial{x}}(x,z) \geq 0$, $\forall z\in \mathbb R$. 
%
%
%\end{proof}

The result of Proposition \ref{prop1_f(x,z)} is interesting. When the fiscal policy has a nonlinear impact on GDP growth, the government will select the primary surplus in order that the effect on GDP matches that on debt, in general (see equation \eqref{eqn:hatu_match}). When this is not possible, the extreme policies are chosen. The choice between $-U_1$ and $U_2$ and the corresponding economic interpretation crucially depend on the government objective, see the comments below Proposition \ref{prop:f(x,z)}.\\

In the case  of Example \ref{example:g}, that is  when the fiscal policy has a linear impact on the GDP growth rate, (see equation \eqref{linear1}), Proposition \ref{prop1_f(x,z)} does not apply, because the condition $\frac{\partial^2{H}}{\partial{u}^2}(x,z,u)>0$, $\forall (x,z,u) (0,+\infty)\times\mathbb{R}\in[-U_1,U_2]$ is not fulfilled.
In the next proposition we discuss a tailor-made result for this special case.

\begin{prop}\label{prop1bis_f(x,z)}
Suppose that the cost function $f$ does not depend on $u$, i.e. $f(x,z,u) = f(x,z)$ and 
let the GDP growth rate given by 
$$ g(z,u) = g_0(z) - \alpha(z) u \;,$$
with $g_0\colon\mathbb{R}\to\mathbb{R}$ and $\alpha(z)>0$, $\alpha(z) \neq 1$, $\forall z \in \mathbb R$.
Then the minimizer of  \eqref{eqn:min_pb} is given by
\begin{equation}
\label{eqn:u*z}
u^*(x,z) = 
\begin{cases} -U_1
& (x,z)\in\Set{(x,z)\in(0,+\infty)\times\mathbb{R} \mid  
	 \frac{\partial{v}}{\partial{x}}(x,z) (\alpha(z) -1) \geq 0 \bigr] }
	\\
	U_2 
	& (x,z)\in\Set{(x,z)\in(0,+\infty)\times\mathbb{R}  \mid   \frac{\partial{v}}{\partial{x}}(x,z) (\alpha(z) -1)  < 0 \bigr] }.
	\end{cases}
\end{equation}
\end{prop}
\begin{proof}
The statement follows by observing that $H$ is a linear function on $ u \in[-U_1,U_2]$
$$H(x,z,u) = -x[g_0(z)  + (1-  \alpha(z)) u]\frac{\partial{v}}{\partial{x}}(x,z) + f(x,z) \;. $$
\end{proof}

In the next proposition we discuss the case where the cost function $f(x,z)$ does not depend on $u$ and it is increasing  in $x \in (0, + \infty)$. This situation refers to the case where debt generates a disutility for the  government of the country and thus it aims to reduce the debt-to-GDP ratio. We refer to this case as the debt reduction problem, see Section \ref{debt-red}.

\begin{prop}
\label{prop:f(x,z)}
Suppose that the cost function $f$ does not depend on $u$, i.e. $f(x,z,u) = f(x,z)$, with $f(x,z)$ increasing  in $x \in (0, + \infty)$. 
%Then the minimizer of equation \eqref{eqn:min_pb} is given \ourcomment{ for any $(x,z) \in (0, + \infty) \times \mathbb R$ such that $\frac{\partial{v}}{\partial{x}}(x,z) \neq 0$} by 
%\[u^*(z) = \argmax_{u\in[-U_1,U_2]}\{ g(z,u) + u \} \;.\]
Assuming $\frac{\partial{v}}{\partial{x}}(x,z) \neq 0$ for any $(x,z) \in (0, + \infty) \times \mathbb R$ and  $\frac{\partial^2{g}}{\partial{u}^2}(z,u) < 0$ for any $(z,u) \in \mathbb R \times [-U_1,U_2]$, the unique minimizer of \eqref{eqn:min_pb} is given by
\[
u^*(z) = -U_1 \lor \hat{u}(z) \land U_2 \;,
\]
where $\hat{u}(z)$ is the unique solution to
\[
\frac{\partial{g}}{\partial{u}}(z,u) = -1, \quad \forall z \in \mathbb{R} \;.
\]
Precisely, $u^*(z)$ admits the following structure:
\begin{equation}
\label{eqn:u*z}
u^*(z) = 
\begin{cases}
	-U_1
	& z\in\Set{z\in\mathbb{R}\mid\frac{\partial{g}}{\partial{u}}(z,-U_1) < -1}
	\\
	U_2 
	& z\in\Set{z\in\mathbb{R}\mid\frac{\partial{g}}{\partial{u}}(z,U_2) > -1}
	\\
	\hat{u}(z) & \text{otherwise} \;.
\end{cases}
\end{equation}
\end{prop}
\begin{proof} 
By Proposition \ref{INCR} we have that  $v(x,z)$ is increasing in $x \in (0, + \infty)$ for any $z\in \mathbb R$, hence 
$\frac{\partial{v}}{\partial{x}}(x,z) \geq 0$, $\forall z\in \mathbb R$, which together with the assumptions 
$\frac{\partial{v}}{\partial{x}}(x,z) \neq 0$, for any $(x,z) \in (0, + \infty) \times \mathbb R$ and  $\frac{\partial^2{g}}{\partial{u}^2}(z,u) < 0$ for any $(z,u) \in \mathbb R \times [-U_1,U_2]$ imply
$$\frac{\partial^2 {H}}{\partial{u}^2}(x,z,u) =
 - x\frac{\partial{v}}{\partial{x}}(x,z)
 \frac{\partial^2{g}}{\partial{u}^2}(z,u) > 0 \quad 
 	\forall (x,z,u) \in (0, + \infty) \times \mathbb R \times [-U_1, U_2] \;.$$
Recalling that 
$$\frac{\partial{H}}{\partial{u}}(x,z,u) =
 - x\frac{\partial{v}}{\partial{x}}(x,z)\biggl[ \frac{\partial{g}}{\partial{u}}(z,u) + 1 \biggr] \;, $$ 
 the first order condition reads as 
 if and only if 
 $$\frac{\partial{g}}{\partial{u}}(z,u) + 1 =0 \;, $$
 and the proof follows  the same lines of Proposition \ref{prop:optimal_Hconvex}. 
\end{proof}

The previous result has an intermediate case as in Proposition \ref{prop1_f(x,z)}, with the same interpretation. However, as announced, now we can provide a deeper insights on extreme fiscal policies.\\
The maximum surplus will be applied only if the beneficial impact on debt more than compensates the negative effect on GDP growth. Indeed, the marginal impact of $U_2$ on the GDP growth is measured by $\frac{\partial{g}}{\partial{u}}(z,-U_1)$, which is negative, while the effect on debt is unitary. When the debt can be decreased more than the GDP growth by means of the maximum surplus, $U_2$ is optimal. Similarly, the maximum deficit is chosen when the positive effect on GDP exceeds the negative effect on debt.\\

In the case of Example \ref{example:g}, that is  when the fiscal policy has a linear impact on the GDP growth rate, (see equation \eqref{linear1}) we have an analogous result.

\begin{prop} \label{prop:linear}
Let the GDP growth rate given by 
$$ g(z,u) = g_0(z) - \alpha(z) u \;,$$
with $g_0\colon\mathbb{R}\to\mathbb{R}$ and $\alpha(z)>0$, $\alpha(z) \neq 1$, $\forall z \in \mathbb R$. Then for any running cost function $f(x,z)$ with $f(x,z)$ increasing  in $x \in (0, + \infty)$ $\forall z \in \mathbb R$, assuming $\frac{\partial{v}}{\partial{x}}(x,z) \neq 0$, for any $(x,z) \in (0, + \infty) \times \mathbb R$,  the unique minimizer of \eqref{eqn:min_pb} is given by

\begin{equation}\label{eqn:u*z1}
u^*(z) =  
\begin{cases}
	-U_1
	& z\in\Set{z\in\mathbb{R}\mid \alpha(z)>1}
	\\
	U_2 
	& z\in\Set{z\in\mathbb{R}\mid  \alpha(z)<1} \;.
\end{cases}
\end{equation}
\end{prop}

\begin{proof}
Observing that 
$$
H(x,z,u) = -x[g_0(z)+ (1 - \alpha(z))u]\frac{\partial{v}}{\partial{x}}(x,z) + f(x,z) \; ,
$$
we have only two cases.
\begin{enumerate}
\item If  $z\in\Set{z\in\mathbb{R}\mid \alpha(z)>1}$ then $H$ is increasing in $u\in[-U_1,U_2]$ thus the minimizer is $u^*(z) =-U_1$.
\item If $z\in\Set{z\in\mathbb{R}\mid  0<\alpha(z)<1}$ then $H$ is decreasing in $u\in[-U_1,U_2]$ thus the minimizer is $u^*(z) =U_2$.
\end{enumerate}
\end{proof}

In the semplified model of  equation \eqref{linear}, that is when $ \alpha(z) =\alpha$ $\forall z \in \mathbb{R}$, we have the following result. 

\begin{corollary}\label{corollary}
Let the GDP growth rate be given by 
$$ g(z,u) = g_0(z) - \alpha u \;,$$
with $g_0\colon\mathbb{R}\to\mathbb{R}$ and $\alpha>0$, $\alpha\neq 1$. Then for any running cost function $f(x,z)$ increasing in $x \in (0, + \infty)$ $\forall z \in \mathbb R$,  assuming $\frac{\partial{v}}{\partial{x}}(x,z) \neq 0$ for any $(x,z) \in (0, + \infty) \times \mathbb R$,  the unique minimizer of \eqref{eqn:min_pb} is constant and given by

\[
u^* = 
\begin{cases}
	-U_1
	& \text{if } \alpha>1
	\\
	U_2 
	& \text{if } 0<\alpha<1 \;.
\end{cases}
\]
\end{corollary}
\begin{proof}

This is a simple application of Proposition \ref{prop:linear}. \end{proof} 

 \begin{remark}\label{ReC}
If $\alpha>1$ we get that the candidate optimal strategy is $u^*= - U_1$, that is  the optimal choice for the government is to generate the maximum deficit. Indeed, the beneficial effect on GDP exceeds the debt increase. By Remark \ref{sustU} this strategy is  sustainable if
 $$U_1 > - {\min_{z \in \mathbb R}  g_0(z) \over \alpha -1}.$$
It would be natural to assume that $\min_{z \in \mathbb R}  g_0(z)<0$, and hence the right hand side is positive. Then $u^*= - U_1$ would be sustainable if the government can produce enough deficit, which is usually the case. \\ 
Conversely, if $\ 0< \alpha <1$, the candidate optimal strategy is $u^*= U_2$, that is the maximum surplus. By Remark \ref{sustU} this strategy is sustainable if 
$$U_2 > - {\min_{z \in \mathbb R}  g_0(z) \over 1- \alpha} \;. $$
If $\min_{z \in \mathbb R}  g_0(z)<0$, the optimal policy $u^*= U_2$ will be sustainable only if the government has the possibility of increasing taxes and reducing public spending more than a given threshold. In some cases, this could be a challenging task for the government.\\
Clearly,  when $\min_{z \in \mathbb R}  g_0(z)>0$, the optimal strategy is always sustainable.
%Clearly if $\min_{z \in \mathbb R}  g_0(z) \geq 0$ the above inequalities are both satisfied, otherwise when the GDP growth rate (with no intervention) is negative  we  get a bound for the level $U_i, i=1,2$,  in order to have the optimal strategy sustainable in both the cases $\alpha>1$ and $0<\alpha<1$.
 \end{remark}

\begin{remark}\label{nof}
Propositions \ref{prop:f(x,z)}, \ref{prop:linear} and Corollary \ref{corollary} show  that in the case where $f(x,z,u) = f(x,z)$ is increasing in $x \in (0, + \infty)$,$\forall z \in \mathbb R$, the minimizer of \eqref{eqn:min_pb} does not depend on the form of the cost function but only on the function $g(z,u)$ which describes the GDP growth rate and the effect of the government policy on it. Thus, given the unitary impact of the fiscal policy on the public debt, the debt-GDP reduction is driven by the effect on the GDP growth rate, which should be the focus of the government attention.
\end{remark}

%--------------------------------------------------------------------------------
%	APPLICATIONS
%--------------------------------------------------------------------------------

\section{Explicit solutions for some cases of interest}\label{Applications}
\label{sec:applications}

In this section we discuss two examples which can be solved applying the Verification Theorem.  In the first example we have in mind a country with debt problems, aiming to reduce its debt ratio. In the second application we discuss the debt smoothing problem, that is, the government wishes to smooth the public debt by flattening its deviation from a given threshold. In the first case the running  cost is increasing w.r.t. $x \in (0,+\infty)$, while in the second case we do not require a monotonic condition.

\subsection{The debt reduction problem}\label{debt-red}

We assume that the cost function is given by
\begin{equation}
\label{eqn:power}
f(x,z) = C(z) x^m, \quad m\geq 2 \;,
\end{equation}
where $C\colon\mathbb{R}\to[0,+\infty)$ is a bounded function.
This disutility function generalizes the quadratic function that is widely used in Economics. The parameter $m$ represents the aversion of the government towards holding debt and the importance of debt for the government is modulated by the function $C(z)$, which is a function of the values of the environment stochastic factor $Z$.

For instance, if $Z$ is an indicator of macroeconomic conditions and higher values correspond to better conditions, assuming that $C$ is increasing enables the government to relax fiscal rules and debt reduction goals when a massive government intervention is needed, as during economic crises. That is, $C(z)$ allows for countercyclical policies.\\

Denoting by 
$$G(z)=\inf_{u\in[-U_1,U_2]}\{-(g(z,u)+u) \} \;,$$ 
the HJB equation \eqref{eqn:HJB_extended} reads as
\[
\begin{split}
&(G(z)+r(z))x\frac{\partial{v}}{\partial{x}}(x,z) + f(x,z)
+ \frac{1}{2}\sigma^2x^2\frac{\partial^2{v}}{\partial{x^2}}(x,z)  \\
&+\rho\sigma x\sigma_Z(z)\frac{\partial^2{v}}{\partial{x}\partial{z}}(x,z) 
+b_Z(z)\frac{\partial{v}}{\partial{z}}(x,z) 
+ \frac{1}{2}\sigma_Z(z)^2\frac{\partial^2{v}}{\partial{z^2}}(x,z) - \lambda v(x,z) = 0 \;.
\end{split}
\]
 With the ansatz $v(x,z)=\phi(z) x^m$ this equation reduces to the following ordinary differential equation (ODE)
 \begin{equation}
\begin{split} \label{ODE}
&\frac{1}{2}\sigma_Z^2(z) \phi'' (z) +  (b_Z(z) + \rho\sigma \sigma_Z(z) m ) \phi' (z) \\  
&+[(G(z)+r(z))m - \lambda + \frac{\sigma^2}{2} m(m-1) ] \phi(z) + C(z)=0 \;.
\end{split} 
\end{equation}
We have the following verification result.

\begin{prop}
Let us assume  there exists a strictly positive and bounded solution  $\phi \in C^2(\mathbb R)$ to equation\eqref{ODE}.
Then the value function is given by $v(x,z)=\phi(z) x^m \in C^{2,2}((0, + \infty) \times \mathbb R)$.
Moreover, assuming  $\frac{\partial^2{g}}{\partial{u}^2}(z,u) < 0$ for any $(z,u) \in \mathbb R \times [-U_1,U_2]$ (or $g$ satisfying \eqref{linear1}) then
$u^*=\{u^*(Z_t)\}_{t\ge0}$ with $u^*(z)$ given in \eqref{eqn:u*z} (\eqref {eqn:u*z1}, respectively) is an optimal strategy.
\end{prop} 

\begin{proof}
%Observing that  $v(x,z)=\phi(z) x^m \in C^{2,2}((0, + \infty) \times \mathbb R)$ satisfies all the assumptions required in Theorem \ref{verification}.
The thesis follows by Theorem \ref{verification} and Proposition \ref{prop:f(x,z)}  (Proposition  \ref{prop:linear}, respectively) 
observing that $\frac{\partial{v}}{\partial{x}}(x,z) = (m-1) \phi(z) x^{m-1}>0$, $\forall  (x,z)  \in (0, + \infty) \times \mathbb R$.  
\end{proof}

In order to obtain an explicit formula for the value function we discuss in next example the simplified case without the presence of the stochastic factor $Z$.

\begin{example}
Let us assume $r(z)=r$, $g(z,u) = g(u)$, $\sigma_Z = b_Z = 0$, $C(z) = C>0$, hence $G(z)=G = \inf_{u\in[-U_1,U_2]}\{-(g(u)+u) \}$ and the HJB  reads as

\begin{equation}
(G+r)xv'(x) + Cx^m
+ \frac{1}{2}\sigma^2x^2v''(x)  - \lambda v(x) = 0 \;.
\end{equation}

It is easy to find an explicit solution to this equation and by the Verification Theorem  \ref{verification} we get that  the value function is given by $v(x) = k x^m$, with

$$k = \frac{C}{ \lambda - (G+r)m - \frac{1}{2}m(m-1)\sigma^2 }>0 \;. $$

As observed in Remark \ref{nof}, we get that the optimal strategy depends only on  the form of the function $g(u)$.  We discuss below three cases.
\begin{enumerate}
\item If $g''(u) <0$ $\forall u \in [-U_1,U_2]$, then by Proposition \ref{prop:f(x,z)}  the optimal control $u^* =  -U_1 \lor \hat{u} \land U_2 $ where $\hat{u}$ the unique solution to $g'(u) = -1$ and precisely  is given by
\begin{equation}
u^*= 
\begin{cases}
	-U_1
	& \text{if}  \; g'(-U_1) < -1
	\\
	U_2 
	& \text{if} \; g'(U_2) > -1
	\\
	\hat{u} & \text{otherwise} \;.
\end{cases}
\end{equation}

\item In the case of \eqref{linear}, $g(u) =g_0 - \alpha u$, the optimal strategy is given by 
\begin{equation}
u^*= 
\begin{cases}
	-U_1
	& \text{if}  \; \alpha >1
	\\
	U_2 
	& \text{if} \;  0<\alpha <1\;.
\end{cases}
\end{equation}
The considerations in Remark \ref{ReC} apply and the optimal strategy is sustainable, when $\alpha >1$ if $U_1 > - {g_0 \over \alpha -1}$ and when $0<\alpha <1$ if $U_2> - {g_0 \over1 -\alpha}$. 

\item$g(u) =g_0$, that is the GDP growth rate is not influenced by the fiscal policy, then the optimal strategy $u^*= U_2$ (the government apply the maximum level of surplus) and it is sustainable if $U_2 > - g_0$. 

\end{enumerate} 
 \end{example}

\subsection{Debt smoothing}

In the previous section we assumed that a cost is associated to any increase of debt, because the aim was to reduce debt. However, in some cases an increase of debt could be more beneficial than its reduction. In \cite{Ostryetal} there is a noteworthy discussion of this topic, focusing on the trade-off between lowering public debt and building public infrastructure. Clearly, the latter should be the greater priority for countries with low debt and big infrastructure needs. Also, there are some unclear cases with high debt, no plausible risk of fiscal distress and some infrastructure needs. In these cases the optimal debt level is unclear, but sometimes the benefit from increasing debt could surpass the advantages of debt reduction. This is even more true at the current juncture, given the very low level of real interest rates and the existence of demand shortfalls. \\
Another reason for the government to not let the public debt fall down to zero is related to the role of debt as savings absorber, see \cite{domar1944}. The government wishes to guarantee at any time a given amount of bonds in order to absorb the private savings in the financial market.\\
Finally, it is well known that sudden and large shocks such as economic crises, wars or pandemics might cause spikes of public debt. ``In developed countries [...] the aim of debt management was to smooth as much as possible the impact of such temporary expenditure shocks that were initially financed by raising debt'' (see \cite{CASALIN2019}). \\

In this section we address the debt smoothing problem. We assume that the government wishes to smooth the public debt by flattening its deviation from a given threshold $\bar{x}>0$. Hence the cost is represented by the (quadratic) distance between the current debt and the target debt $\bar{x}$.  Formally, we assume that the cost function is given by
\begin{equation}
\label{eqn:quadratic_1}
f(x) = (x- \bar{x})^2 \;.
\end{equation}

Recalling Proposition \ref{prop:v_properties} and Proposition \ref{prop:vconvex}, we can derive the following properties of the value function $v(x,z)$: 
 \begin{itemize}
\item $v(x,z) \geq 0$, $\forall (x,z) \in (0, + \infty)\times \mathbb{R}$\,;
\item $v(x,z)$ is convex w.r.t $x \in (0, + \infty)$ $\forall z \in \mathbb{R}$\,;
\item $v(x,z) \leq {x^2\over \lambda - \lambda_2} + {\bar{x}^2\over \lambda}$ $\forall (x,z) \in (0, + \infty)\times \mathbb{R}$\,;
\item $\lim_{x\to 0^+} v(x,z) \leq {\bar{x}^2\over \lambda}$ $\forall z \in \mathbb{R}$. 
\end{itemize}

In order to provide an explicit solution to the HJB equation we discuss the simplified case without the stochastic factor. Precisely, we assume $\sigma_Z=b_Z=0$, $r(z)=r$ and $g(z,u) = g(u)$. Then we denote $G_1 = min_{u\in[-U_1,U_2]}\{-(g(u)+u) \}$ and $G_2 = max_{u\in[-U_1,U_2]}\{-(g(u)+u) \}$, thus the HJB \eqref{eqn:H} reads as
\begin{equation}\label{1}
(G_1+r)xv'(x) + (x- \bar{x})^2
+ \frac{1}{2}\sigma^2x^2 v''(x)  - \lambda v(x) = 0 \;,
\end{equation}
for $v'(x) \geq 0$, and 
\begin{equation}\label{2}
(G_2+r)xv'(x) + (x- \bar{x})^2
+ \frac{1}{2}\sigma^2x^2 v''(x)  - \lambda v(x) = 0 \;,
\end{equation}
for $v'(x) < 0$. Recalling that $v'$ is increasing (because $v$ is convex), we conjecture that there exists a threshold $ \widetilde x>0$ such that $v'(x)>0$ $\forall x >  \widetilde x$ and $v'(x)<0$ $\forall x <  \widetilde x$, so that we make the following ansatz: 
\begin{equation} \label{v}
v(x)= 
\begin{cases}
	a_1x^2 + b_1 x+ c_1  + d_1 x^{\gamma_1}, 
	& \text{if}  \; x \geq \widetilde x
	\\
	a_2x^2 + b_2 x+ c_2  + d_2 x^{\gamma_2}, 	
	& \text{if} \; x < \widetilde x \;.
\end{cases}
\end{equation}
By substituting this expression in the equations \eqref{1} and \eqref{2} above, we get that
\begin{equation}\label{ai}
 a_i  = {1\over \lambda - 2(G_i +r) - \sigma^2}, \;  b_i = {2  \bar{x} \over G_i +r - \lambda}, \; c_i = {\bar{x}^2 \over \lambda}, \; i=1,2 \;,
\end{equation}
\begin{equation}\label{gamma1}
	\gamma_1 = {-[2(G_1+r) - \sigma^2] - \sqrt{ [2(G_1+r) -\sigma^2]^2 + 4 \lambda \sigma^2 }
	\over 2 \sigma^2} <0 \;,
\end{equation}
\begin{equation}\label{gamma2}
	\gamma_2 = {- [2(G_2+r) - \sigma^2] + \sqrt{ [2(G_2+r) -\sigma^2]^2 + 4 \lambda \sigma^2  } 
	\over 2 \sigma^2} >0 \;.
 \end{equation}

Condition \eqref{gamma1} guarantees  the quadratic growth, that is $v(x) \leq C(1 + x^2)$ %{x^2\over \lambda - \lambda_2} + {\bar{x}^2\over \lambda}$, 
 $\forall x \in (0, + \infty)$, while condition \eqref{gamma2}  implies  that $\lim_{x\to 0^+} v(x) = {\bar{x}^2\over \lambda}$.

We also conjecture that $v$ is twice continuously differentiable. Then the three constants  $d_i, i=1,2$ and $\widetilde x$ can be found by taking the following conditions into account:
\begin{enumerate}
\item $v(\widetilde x^+) = v(\widetilde x^-) $;
\item $v'(\widetilde x^+) = v'(\widetilde x^-) $;
\item $v''(\widetilde x^+) = v''(\widetilde x^-)$.
\end{enumerate}
In view of \eqref{v} these equations read as:
\begin{enumerate}
\item $a_1 \widetilde x^2 + b_1 \widetilde x + d_1 \widetilde x^{\gamma_1} = a_2 \widetilde x^2 + b_2 \widetilde x+  d_2 \widetilde x^{\gamma_2}$;
\item $2a_1 \widetilde x + b_1  + \gamma_1 d_1 \widetilde x^{\gamma_1 -1} = 2 a_2 \widetilde x + b_2 + d_2 \gamma_2 \widetilde x^{\gamma_2 -1} =0$;
\item $2a_1 +  \gamma_1 (\gamma_1 -1) d_1 \widetilde x^{\gamma_1 - 2} = 2 a_2 + d_2 \gamma_2 (\gamma_2 -1 ) \widetilde x^{\gamma_2 -2}$. 
\end{enumerate}

Conditions 2 and 3 are  equivalent to
 \begin{equation}\label{d1} 
 d_1 = d_1(\widetilde x)=  - {2a_1 \widetilde x+ b_1 \over \gamma_1 \widetilde x^{\gamma_1 -1}}
 \end{equation}
%(segue da $v'(\widetilde x^+) = 2a_1 \widetilde x + b_1  + \gamma_1 d_1 \widetilde x^{\gamma_1 -1}=0$)
and
\begin{equation}\label{d2} d_2 = d_2(\widetilde x)= {2(a_1-a_2) \widetilde x  - (\gamma_1 -1)(2a_1 \widetilde x + b_1)  \over \gamma_2 (\gamma_2 -1 ) \widetilde x^{\gamma_2 - 3}} \;,
\end{equation}
respectively. Then the equation of item 1 can be rewritten as
\begin{equation}\label{x tilde}
a_1 \widetilde x^2 + b_1 \widetilde x + d_1(\widetilde x) \widetilde x^{\gamma_1} = a_2 \widetilde x^2 + b_2 \widetilde x+  d_2(\widetilde x) \widetilde x^{\gamma_2} \;.
\end{equation}
Clearly, we first need to compute $\widetilde x$ by solving numerically equation \eqref{x tilde}, then we obtain $d_1$ and $d_2$ by equations \eqref{d1} and \eqref{d2}, respectively. \\

Applying the Verification Theorem \ref{verification} we obtain the following result.

\begin{prop}
\label{prop:debtsmooth}
Let $a_i$, $b_i$, $c_i$, $\gamma_i$, $i=1,2$ given in equations \eqref{ai} - \eqref{gamma1} - \eqref{gamma2},
$\widetilde x \in (0, + \infty)$ solution to equation \eqref{x tilde}, and $d_1$ and $d_2$ two constants determined by \eqref{d1} and \eqref{d2}, respectively. Define 
 \begin{equation} \label{V}
w(x)= 
\begin{cases}
	a_1x^2 + b_1 x+ c_1  + d_1 x^{\gamma_1}, 
	& \text{if}  \; x \geq \widetilde x
	\\
	a_2x^2 + b_2 x+ c_2  + d_2 x^{\gamma_2}, 	
	& \text{if} \; x < \widetilde x \;.
\end{cases}
\end{equation}

If $\forall x \in (0, + \infty)$ $w''(x) >0$ and (i) $g''(u) \neq 0$ $\forall u \in [-U_1,U_2]$, or (ii) $g$ is given by \eqref{linear}, then $w$ is the value function (i.e. $w=v$) and $u^* = \{u^*(X^{u^*}_t)\}_{t\geq 0}$ is an optimal strategy, where the function $u^*(x)$ is given by: \\

Case (i)  
\begin{equation}  
u^*(x) = 
\begin{cases} -U_1
& x \in\Set{x\in(0,+\infty) \mid  0
	\ge w'(x)(g'(-U_1) + 1) }
	\\
	U_2 
	& x \in\Set{x\in(0,+\infty) \mid  0 \le w'(x)(g'(U_2) + 1) }	\\
	\hat{u} & \text{otherwise} \;,
\end{cases}
\end{equation}
where $\hat{u}$ is the unique solution to
\[
g'(u) = -1 \quad u \in [-U_1, U_2]\;.
\]

Case (ii)  
\begin{equation} \label{u_ex}
 u^*(x) = 
\begin{cases} -U_1
& x \in\Set{x\in(0,+\infty) \mid  w'(x)(\alpha -1) \geq 0 }
	\\
	U_2 
	& x \in\Set{x\in(0,+\infty) \mid   w'(x)(\alpha -1) <0 } \;.
	\end{cases}
\end{equation}
\end{prop}
\begin{proof}
It is sufficient to show that all the conditions of the Verification Theorem \ref{verification} are satisfied. 
First we prove that  $w$ is twice continuous differentiable, by construction we have that  $w(\widetilde x^+) = w(\widetilde x^-)$, $w''(\widetilde x^+) = w''(\widetilde x^-)$ and $w'(\widetilde x^+)=0$.  Let us observe that $w'(\widetilde x^-)= w'(\widetilde x^+)=0$ because $w$ is twice differentiable on $(0, + \infty)$.
Since $w'$ is strictly  increasing, we have that $w'(x) \geq 0$ for $x \geq  \widetilde x$ and
 $w'(x) <0$  for $x <\widetilde x$. $w$ solves the HJB equation \eqref{eqn:H} and satisfies the quadratic growth condition by construction.  Then the statement follows by Propositions \ref{prop1_f(x,z)}, \ref{prop1bis_f(x,z)} and the Verification Theorem \ref{verification}.
  \end{proof} 

\begin{remark}
\label{rem:ubinaria}
Let us observe that \eqref{u_ex} reads as:
\begin{itemize}
\item for $0 < \alpha <1$:
\begin{equation} \label{u_ex1}
 u^*(x) = 
\begin{cases}U_2
& \text{if} \,  x \geq \widetilde x
	\\
	- U_1 
	& \text{if} \,  x < \widetilde x \;;
	\end{cases}
\end{equation}
\item for $\alpha > 1$
 \begin{equation} \label{u_ex2}
 u^*(x) = 
\begin{cases} -U_1
& \text{if} \, x \geq \widetilde x
	\\
	U_2 
	& \text{if} \, x < \widetilde x \;.
	\end{cases}
\end{equation}
\end{itemize}
We can distinguish two cases. When the impact of fiscal policy on GDP growth is low (i.e. $0<\alpha<1$), the government applies the maximum surplus, increasing taxes and decreasing the public spending, as long as the current debt-GDP ratio is over the threshold $\widetilde x$. When $X_t$ is below $\widetilde x$, the maximum deficit is applied. This simple rule is reversed when $\alpha>1$, that is when the fiscal policy is more effective on GDP growth than on debt.\\
%\nuovo{COMMENTO DI PRIMA: 
%This means that the optimal policy for the government is described as follows:
%for $0 < \alpha <1$ the government intervenes with the maximum level  of deficit when the actual debt ratio is above the threshold $\widetilde x$
%and and applies the maximum level of taxes  when the actual debt ratio is below the threshold $\widetilde x$ and the opposite for $\alpha \geq 1$.
%Therefore the computation of the  threshold $\widetilde x$ is an essential task for a country.}
\end{remark}

In the sequel we perform some numerical simulations to further investigate the results of Proposition \ref{prop:debtsmooth}. In particular, we refer to the case (ii), when the GDP growth rate takes the form of equation \eqref{linear}, that is 
\[
g(z,u) = g_0 - \alpha u \;.
\]
We consider the parameters in Table \ref{tab:parameters} below as a reference scenario, unless otherwise specified.

\begin{table}[htp]
\caption{Simulation parameters}
\label{tab:parameters}
\centering
\begin{tabular}{ll}
\toprule
\textbf{Parameter} & \textbf{Value}\\
\midrule
$r$ & $0.01$ \\
$g_0$ & $0.03$\\
$\sigma$ & $0.2$\\
$\bar{x}$ & $0.6$\\
$U_1=U_2$ & $1$\\
$\alpha$ & $0.9$\\
$\lambda$ & $5$\\
\bottomrule
\end{tabular}
\end{table}

The main task is to solve numerically \eqref{x tilde}. Then we can investigate how the result is sensitive to the model parameters.\\

For instance, given the parameters as in Table \ref{tab:parameters}, we can find a solution to equation \eqref{x tilde}, that is $\tilde{x}=0.6194$, so that the value function is well defined and it is illustrated in Figure \ref{img:vfun}. The reader can easily notice that it is convex, as expected.

\begin{figure}[htp]
\centering
\includegraphics[width=\textwidth]{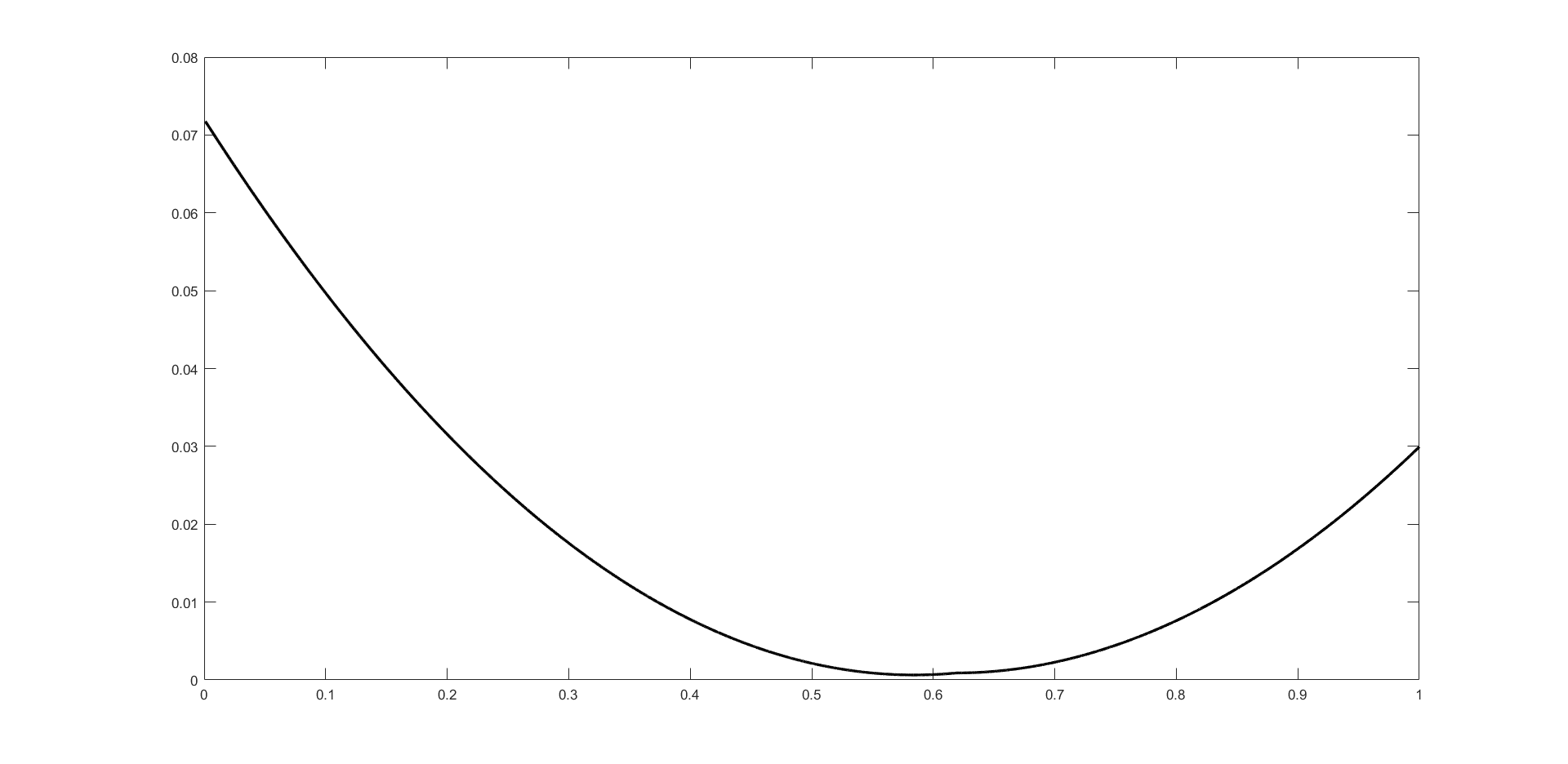}
\caption{The value function given by equation \eqref{V}.}
\label{img:vfun}
\end{figure}

Recalling \eqref{eqn:X0} we get that the uncontrolled debt ratio is given by  
$$X^{0,x}_t = xe^{(r-g_0 - \frac{1}{2}\sigma^2) t   + \sigma W_t}, \quad t \geq 0 \;,$$
hence we  call strong economy countries with parameters  satisfying  $g_0 + {\sigma^2\over 2} > r$ and weak economy when the opposite  inequality holds, that is  $g_0 + {\sigma^2\over 2} < r$ (e.g. in \cite{CA2018} the same definition is used).  

In Table \ref{tab:strongweak} we perform a comparison between countries with different economy parameters; precisely, we consider the parameters $r=0.01, g_0=0.03, \sigma=0.2$ corresponding to a strong economy and $r=0.07, g_0=0.015, \sigma=0.3$ for a weak economy. 

\begin{table}[htp]
\caption{Comparison between strong and weak economies}
\label{tab:strongweak}
\centering
\begin{tabular}{lll}
\toprule
\textbf{Parameters} & \textbf{Strong Economy} & \textbf{Weak Economy}\\
\midrule
$\alpha = 0.9$ & $\tilde{x}=0.6194$ & $\tilde{x}=0.6241$\\
$\alpha = 0.95$ & $\tilde{x}=0.6052$ & $\tilde{x}=0.5941$\\
$U_1=U_2=0.8$ & $\tilde{x}=0.6125$ & $\tilde{x}=0.6068$\\
$U_1=U_2=0.5$ & $\tilde{x}=0.6052$ & $\tilde{x}=0.5941$\\
\bottomrule
\end{tabular}
\end{table}

First, we noticed that the solution $\tilde{x}$ is symmetric with respect to $\alpha=1$, i.e. $\tilde{x}$ for $\alpha=0.9$ and $\alpha=0.95$ is the same as for $\alpha=1.1$ and $\alpha=1.05$, respectively. We can observe that when the effect of fiscal policy on GDP, which is measured by $\alpha$, is closer to the effect on debt, which is a unitary, then the threshold $\tilde{x}$ decreases. In particular, for the weak economy this phenomenon is more evident and $\tilde{x}$ becomes lower than the target debt $\bar{x}$. Moreover, when the fiscal margin is reduced ($U_1$ and $U_2$ go from $1$ to $0.8$ and $0.5$) we can see that $\tilde{x}$ decreases, becoming closer to $\bar{x}$. Hence, when the fiscal policy is less effective, the government should intervene even when minimal deviations from the target are observed.\\

To clarify the government behavior and give a deeper insight into Remark \ref{rem:ubinaria}, we simulate a trajectory of debt-GDP ratio with the corresponding optimal strategy. 

Figure \ref{img:Xcontrolled} illustrates a trajectory of the (optimally) controlled Debt-GDP. Recall that in this case $0 < \alpha<1$, hence the optimal strategy follows equation \eqref{u_ex1}. Let us recall that here we have $\tilde{x}=61.94\%$, while our target is $\bar{x}=60\%$. For instance, let us comment the first two months as illustrated in the picture. Starting from a debt-GDP ratio of $70\%$, the government applies the maximum surplus $U_2$ according to the optimal strategy. We observe the debt decreasing and after 51 days the debt hits the threshold $\tilde{x}$, so that the maximum deficit $-U_1$ is applied. Then the debt increases again and the government will stop spending when the threshold is hit again.

\begin{figure}[htp]
\centering
\includegraphics[width=\textwidth]{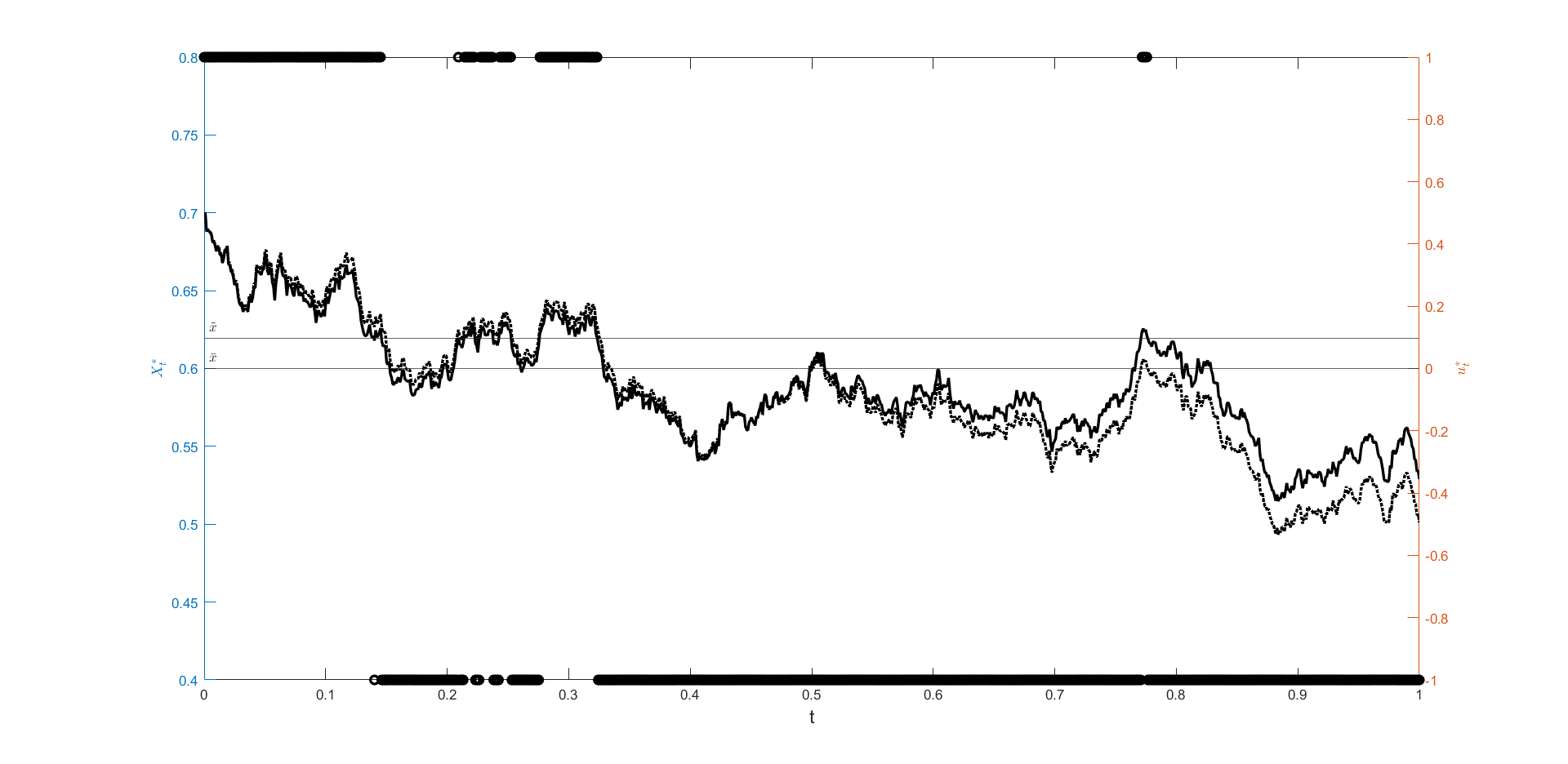}
\caption{A simulation of the Debt-GDP ratio (left y axis) with the corresponding optimal strategy (right y axis). Here $\tilde{x}=61.94\%$ and $\bar{x}=0.60\%$. The solid line represents the controlled process, while the dotted line is the Debt-GDP ratio without intervention.}
\label{img:Xcontrolled}
\end{figure}
%
%\begin{figure}[htp]
%\centering
%\includegraphics[width=\textwidth]{Xsimulated2.png}
%\caption{A simulation of the controlled Debt-GDP ratio with the corresponding optimal strategy. Here ...}
%\label{img:Xcontrolled2}
%\end{figure}

\section*{Conclusions}

The present article introduces a new dynamic stochastic model for the debt-GDP ratio, see equation \eqref{eqn:debt}, trying to fill the gap between the economic theory and the mathematical model formulation. In particular, the expansionary/recessionary role of the fiscal policy on the GDP were not fully recognized by some recent research papers, while this phenomenon is well known in Economics. %As a result, the effect of the fiscal policy on the debt-GDP ratio is not unidirectional, so that a reduction of the instantaneous debt-GDP growth rate can be obtained through surplus or deficit fiscal policies.
\\
We study a general class of debt management problems as a stochastic control problem with a generic functional cost, depending on debt-GDP ratio, an exogenous stochastic factor $Z$ and the fiscal policy. In Section \ref{optimal_stra} we provide a characterization of the value function as a viscosity solution of the HJB equation, see Theorem \ref{thm:valuefun_viscosity}. Moreover, a verification theorem (see Theorem \ref{verification}) is proved when a classical solution of the HJB equation exists. Then in Section \ref{optimal_stra} we investigate the optimal candidate fiscal policy in many cases of interest. Some practical applications have been analyzed in Section \ref{sec:applications}.\\

Our results give many interesting insights on the government interventions for managing public debt issues. When the fiscal policy $u$ has a nonlinear impact on GDP growth, the government will select the primary balance in order that the effect on GDP matches that on debt. When this is not possible, the extreme policies (maximum deficit or maximum surplus) are chosen (see Proposition \ref{prop1_f(x,z)}). When the effect of $u$ on GDP is linear, a bang-bang strategy is obtained. \\
When the cost functional does not depend on $u$ and it is increasing in $x$ (e.g. for debt reduction), we proved that the maximum surplus $U_2$ is applied only when its marginal negative effect on GDP is lower than the positive effect on debt reduction. Similarly, the maximum deficit is optimal when its marginal beneficial effect on GDP surpasses the negative effect of increasing debt (see Proposition \ref{prop:f(x,z)}).\\

In general, it turns out that the impact of the fiscal policy on GDP is crucial to determine the optimal taxation/spending level and the corresponding optimal debt-GDP ratio. In particular, depending on the macroeconomic conditions (described by $Z$) and the magnitude of the impact of the fiscal policy on GDP, the optimal debt management can be achieved by deficit policies in some circumstances. This is still true when the government goal is debt reduction, as shown in Section \ref{debt-red}.

\section*{Declarations}

\textbf{Funding}
The authors work has been partially supported by the Project INdAM-GNAMPA, number: U-UFMBAZ-2020-000791.\\
\textbf{Conflicts of interest/Competing interests}
The authors declare no conflict of interest.\\
\textbf{Availability of data and material}
Not applicable.\\
\textbf{Code availability}
Not applicable.\\
\textbf{Ethics approval}
Not applicable.\\
\textbf{Consent to participate}
Not applicable.\\
\textbf{Consent for publication}
Not applicable.\\

%--------------------------------------------------------------------------------
%	BIBLIOGRAPHY
%--------------------------------------------------------------------------------

\newpage
\bibliographystyle{abbrv}
\bibliography{biblio}

\end{document}